\documentclass[letterpaper]{article} 
\usepackage[preprint]{aaai2027}  
\usepackage[hyphens]{url}  
\usepackage{graphicx} 
\usepackage{natbib}  
\usepackage{caption} 
\usepackage{algorithm}
\usepackage{algorithmic}

\usepackage{newfloat}
\usepackage{listings}
\DeclareCaptionStyle{ruled}{labelfont=normalfont,labelsep=colon,strut=off} 
\floatstyle{ruled}
\newfloat{listing}{tb}{lst}{}
\floatname{listing}{Listing}

\usepackage{booktabs}

\usepackage{amssymb}
\usepackage{amsmath}
\usepackage{subcaption}
\usepackage{tabularx}

\makeatletter
\renewcommand{\subsubsection}{%
  \@startsection{subsubsection}{3}{\z@}%
  {-6pt plus -2pt minus -1pt}%
  {-1em}%
  {\normalsize\bf}%
}

\usepackage{soul}

\newcommand{\stkout}[1]{%
  \ifmmode
    \text{\st{\ensuremath{#1}}}%
  \else
    \st{#1}%
  \fi
}

\usepackage[colorinlistoftodos,textsize=scriptsize]{todonotes} 

\usepackage{amssymb}
\usepackage{amsmath}
\usepackage{amsthm}
\usepackage{booktabs}

\newcommand{\newstarredtheorem}[3][]{%
  \if\relax\detokenize{#1}\relax
    \newtheorem{#2}{#3}%
  \else
    \newtheorem{#2}[#1]{#3}%
  \fi
  \newtheorem*{#2star}{#3}%
  \expandafter\let\csname #2*\expandafter\endcsname
    \csname #2star\endcsname
  \expandafter\let\csname end#2*\expandafter\endcsname
    \csname end#2star\endcsname
}

\newstarredtheorem{theorem}{Theorem}
\newstarredtheorem[theorem]{definition}{Definition}
\newstarredtheorem[theorem]{obs}{Observation}
\newstarredtheorem[theorem]{lemma}{Lemma}
\newstarredtheorem[theorem]{proposition}{Proposition}
\newstarredtheorem[theorem]{corollary}{Corollary}
\newstarredtheorem[theorem]{observation}{Observation}

\newtheorem*{thm*}{Theorem}
\newtheorem*{prop*}{Proposition}

\title{The Limits of Manipulation within Personalized Incentive Policies}
\author{
    Andrea Araldo\textsuperscript{\rm 1},
    Alessa Mayer\textsuperscript{\rm 2},
    Ravi Seshadri\textsuperscript{\rm 3},
    Amir Brudner\textsuperscript{\rm 4,\rm 5}
}
\affiliations{
    \textsuperscript{\rm 1}SAMOVAR, Télécom SudParis, Institut Polytechnique de Paris, 91120 Palaiseau, France\\
    \textsuperscript{\rm 2}Institut Polytechnique de Paris\\
    \textsuperscript{\rm 3}Technical University of Denmark\\
    \textsuperscript{\rm 4}Massachusetts Institute of Technology\\
    \textsuperscript{\rm 5}The Hebrew University Business School
}

\begin{document}

\maketitle

\begin{abstract}
Demand management policies aim to shift individual behavior toward sustainable choices (e.g., goods, transport modes). Personalized incentive policies tailor subsidies to agents’ needs and constraints, but personalization also exposes the regulator to agents' strategic manipulative behavior. The extent to which agents can profitably manipulate such policies remains poorly understood, which limits their applicability.

This paper establishes analytically the conditions under which manipulation benefits the agent. We formalize a Stackelberg game in which an agent may simulate a different profile to extract higher incentives from a regulator. Previous work typically assumes that the regulator would \textit{instantaneously} classify agents into the simulated user profile. In contrast, we explicitly model the regulator the process of learning the agent profile, based on observed past behavior. This makes manipulation harder, since it forces the agent to consistently behave according to a simulated profile, thus renouncing to follow their true preferences. Manipulative agents thus incur a \emph{Cost of Manipulation}, which has been typically been neglected. 
By accounting for it, we get a counterintuitive finding: although classification uncertainty reduce personalization accuracy, it may benefit the regulator by deterring manipulation. We establish an upper bound on the agent’s gain from manipulation, which becomes stricter as profiles become harder to distinguish for the regulator. Simulations on the London Passenger Mode Choice dataset illustrate this deterrence effect. The framework can therefore be used to stress-test the robustness of personalized incentive policies before deployment.
\end{abstract}


\section{Introduction}
Behavioral change remains a major challenge in achieving sustainability goals, especially as technological solutions alone are insufficient. Demand Management Policies (DMP) favor changes in the agents’ behavior toward sustainable choices \cite{crow2021we}. Potential greenhouse gas reduction achievable via DMP is estimated to exceed 40\% by 2050, compared to baselines \cite{lee2023ipcc,creutzig2022demand}. We will consider DMP in which a regulator, e.g., a government or any type of authority, applies a signal. 

Negative signals (e.g., prices, taxes, bans) generally lack public acceptability (to the point of violent protests \cite{dabla2023public}) and are often inadvertently inequitable \cite{Blasco}. For instance, increasing taxes on fuel particularly impacts low-income rural households, who are dependent on their cars to perform their everyday activities, and who spend, for transport, already up to 30\% of their income \cite{lebras2019voiture}. On the other hand, positive signals (e.g., incentives, subsidies) are expensive for the regulator and may not be sufficiently strong to induce relevant behavioral shifts \cite{gneezy2000pay}.


Personalized DMP (PDMP) can help overcome the above issues \cite{javaudin2023personalised,ascarza2022eliminating,XIE2024realtime} by adapting signals to the condition and the preferences of each agent. Personalized signals can target agents for whom a behavioral shift is less costly, thus mitigating inequitable effects. Thus, PDMP have the potential to be more effective and acceptable than non-personalized DMP.

However, PDMP are vulnerable to manipulative agents, who may make dishonest choices in order to conceal their true preferences and manipulate the regulator’s signals in their favor. For instance, it has been shown that some individuals maliciously reduced their labor supply to get a subsidy dedicated to the under-employed \cite{firpo2014evidence}. 

In this work, we focus on this critical vulnerability by modeling the interaction between a regulator and a potentially manipulative agent in a Stackelberg game. The regulator (the leader) estimates agent (follower) profiles based on observed choices, while the agent may adopt a simulated (i.e., fake) utility function to secure higher incentives. We investigate how the regulator's continuous estimation of agent utility impacts manipulation.
Extensive work in artificial intelligence on strategic behavior~(\S{}\ref{sec:related-work}) models agents in ways that are not behaviorally-grounded, and thus not compliant with the way human agents make choices. They are thus not applicable to demand-side mitigation.
Our contribution is the first behaviorally-grounded model of repeated interaction between a regulator and an agent within a demand-side mitigation scenario, consisting of a personalized incentive policy. We model agents' choices, based on Random Utility Theory~(RUT)~\cite{Strzalecki2025,mcfadden1981econometric}. In contrast with previous work, in which the agent simply reports their profile, we explicitly model the regulator's~\emph{preference elicitation}~\cite{chapman2025preference}, with which the regulator learns the agent's profile. 

Previous results showed that the agent could greatly manipulate the regulator. Our analysis shows instead that such findings are over-pessimistic. Indeed, we theoretically show that realistic preference elicitation deters manipulative behavior and, consequently, personalized incentive policies can effectively reduce social cost. Results are validated in Monte Carlo simulations.
%
%
%
%
The proofs that are not in the main paper are in the Supplementary Material.

\section{Related Work}
\label{sec:related-work}

Strategic interaction and manipulation are studied in artificial intelligence, game and learning theory. In Adversarial Interaction Attacks \cite{nguyen2019simultaneous}, a defender allocates resources to defend targets and an attacker chooses the target to attack. The defender seeks to learn the attacker's profile to better allocate defense resources. The attacker can manipulate the defender, by implementing certain sequences of attacks that steer the defender's learning in a way that increases the utility of the attacks. Our case, instead, cannot be simultaneous: first, the regulator must set the incentives for the different available alternatives; then, the agent chooses an alternative. Our setting is thus better described by a Stackelberg game, where the leader (the regulator) takes their action (setting incentives) and \emph{then} the follower (the agent) takes their action (choosing an alternative).

In recent work, when the leader has limited information about the follower, the latter can
manipulate the outcome of the interaction. 
\cite{gan2019manipulating,gan2019imitative,nguyen2019imitative,birmpas2020optimally,dawkins2021limits} show that the follower can implement an optimal manipulative strategy, which gives the follower higher payoff: the leader is thus greatly vulnerable to manipulation, which is simply performed by the follower via strategic reporting. Such previous findings would tend to discourage the application of personalized policies. We however observe that their setting is unrealistic and too pessimistic for the leader. Indeed, manipulation is ``too easy'': ``To manipulate the [leader], the [follower] reports a fake [profile]'', and the leader just considers the reported profile when deciding their action. The previous quote is from \cite{gan2019manipulating}, but this assumption is common to all the aforementioned work. Similarly, framing the setting within Bayesian Persuasion \cite{Kamenica2011,jain:hal-04912547}, an agent could manipulate the regulator by choosing the information to report. 

We consider instead a more realistic setting, in which the regulator (leader) \emph{learns} the agent (follower) profile, by interacting with them and observing their previous choices, performing \emph{preference elicitation} \cite{chapman2025preference}. In the following, we will thus refer to \emph{instantaneous} and \emph{realistic preference elicitation}, to contrast, respectively, the setting of the aforementioned work (where the leader immediately considers the profile reported by the agent when setting incentives) and the more realistic setting of our work (where the leader considers instead the profile they learn by observing the agent's behavior).
Related are screening contracts, with which a regulator could design a menu of options, specifically to distinguish among agent profiles (Einav et al. 2022). Our setting differs because the options (alternatives) are exogenous and agent choices are stochastic and thus do not perfectly reveal their profile.

In \cite{sessa2020learningtoplay} the agent can \emph{adversarially} choose one profile at every timeslot. However, their setting is too pessimistic (and thus inappropriate) for studying personalized policies, for different reasons. First, we assume that agents behave according to Random Utility Theory (RUT) \cite{Strzalecki2025,mcfadden1981econometric}, which is more realistic in the economic settings we are tackling. This assumption constrains the agent's alternatives to be probabilistically chosen, at every timeslot, according to a single profile. The manipulation of the agent thus consists in committing to a ``fake'' profile, instead of adversarially selecting a new profile at every timeslot as in \cite{sessa2020learningtoplay}. Moreover, we assume agents aim to maximize their utility, and not necessarily act adversarially (e.g., aim to maximize the loss of the regulator). Second, \cite{sessa2020learningtoplay} assume the leader adopts probabilistic policies, while in our case incentive policies are constrained to be deterministic. Indeed, it would not be acceptable that a regulator (e.g., a government) gives randomized incentives to individuals. 

\citet{ionescu2023strategic} study strategic behavior in matching algorithms, while we consider demand-side management policies. Moreover, their manipulation is implemented by ``returning agents'', which take the role of our regulator. We instead study user-side manipulation.
A different, but related, problem is the personalized reserved price setting in a repeated auction, considered by \cite{pmlr-v108-nedelec20a}. In their work, however, the set of leaders (corresponding to the buyers) are those who can potentially manipulate the follower (the seller). 
As a consequence, their findings do not hold in our case, in which agents (followers) can manipulate the regulator (leader).

In Performative Prediction (PP) \cite{perdomo2020performative}, predictions produced by a supervised learning model alter the data-generating process itself, thereby changing the distribution on which future predictions are evaluated. However, in PP the action of the leader is to change classification parameters, whereas the action of our leader (the regulator) is to propose incentives. Moreover, different from PP, the goal of our leader is not to predict, but to propose incentives. In PP, the leader aims to reduce classification loss, whereas our leader aims to reduce generalized social cost. In PP, the agent influences the classifier parameters, which is not our case. Also, PP is not applicable to demand-side policies: in \cite[\S{}5.1]{perdomo2020performative}, the regulator is allowed to change the classifier over time. However, for transparency constraints, a regulator, such as a public authority, must fix and announce the rules in advance; they would otherwise be unacceptable and perceived as black-box or arbitrary by citizens. This is why we fix the agent's profile classification is performed via the Maximum Likelihood Estimator, according to real preference elicitation cases \cite{chapman2025preference}.

Our setting has some similarity with the Vickrey–Clarke–Groves (VCG) mechanism. The latter computes prices so as to maximize social welfare, while ensuring truthful behavior of the agents. One could be tempted to use a VCG-like mechanism in our setting. This is however not possible, because in VCG, the social welfare is the sum of the agent's utilities, while in our case the social cost that the regulator aims to minimize is not such a sum, but is rather related to sustainability goals (e.g., emissions).

Despite the rich literature in the analysis of manipulative behavior, the settings considered are either too simplistic - the agent can manipulate by simply misreporting - or too pessimistic - the manipulative behavior is not appropriately constrained to be economically and behaviorally grounded. This paper aims to fill this gap.


\vspace{-0.2cm}
\section{Model}
\label{sec:model}
\vspace{-0.2cm}
Notation is entirely defined in the paper. For ease of reference, it is summarized in Table~\ref{tab:notation} of the the Suppl. Material. We formulate a repeated Stackelberg game (as \cite{birmpas2020optimally}) with a leader (regulator) and a follower (agent). At each iteration, the agent chooses an alternative~$i\in\mathcal I$, which generates externality~$c_i\ge 0$. Let alternatives $\mathcal I=\{1,\dots,I\}$ be sorted in the increasing order of externalities, i.e., from the most to the less sustainable. To encourage the choice of the most sustainable alternative, at the beginning of each iteration the regulator proposes incentive~$y\ge 0$.

To give our model some context, we briefly mention some real-world applications. For instance, consumers (agents) can choose to buy a short-lasting, unethically produced clothing item, which would correspond to alternative~$2$, with~$c_2$ representing the environmental and social externality caused by producing this good. On the other hand, alternative~$1$ would be to buy this clothing item second-hand, which has low externality~$c_1$. Incentive~$y$ would be a subsidy to encourage buying second-hand. Another application consists in the choice of a traveler (agent) of public transport (alternative~$1$) or taking the car (alternative~$2$) for a trip. Here, externality~$c_2$ would represent the car-related externality, e.g. pollution, induced congestion. The government could provide an incentive~$y$ to encourage taking public transport, which produces lesser externality~$c_1$. A last example is the decision of a firm (agent) to prefer electricity generated from renewable energies (alternative~$1$) or traditional, fossil‑fuel‑based electricity (alternative~$2$), which imposes a higher externality~$c_2$ in the form of CO$_2$ emissions. The regulator can set a subsidy~$y$ to encourage the firm to use renewable sources.

\subsection{Random Utility Modeling}
\vspace{-0.1cm}
We model agent behavior according to random utility theory~(RUT) \cite{Strzalecki2025,mcfadden1981econometric}. The agent has a \emph{utility function}~$v(i,y)$, associating a certain \emph{systematic utility} (or, briefly, utility \cite[\S{}3.2]{Cascetta2009}) to alternative~$i$ when the regulator proposes incentive~$y$.
When the agent chooses alternative~$i\in\mathcal{I}$, they get random utility: systematic utility~$v(i,y)$ plus a random variable~$\epsilon_i$, called \emph{random term}. Following RUT, we assume~$\epsilon_i$ are identically and independently drawn from a Gumbel distribution with $0$-mean and scale parameter~$\theta\ge 0$.
The random term accounts for 
some inherent randomness in agent's behavior \cite[p100]{Cascetta2009}, which makes the agent's choice not easily predictable for the regulator. This randomness may be due to the inability of the agent to estimate the exact alternatives' utility or to an additional utility disturbance, which is completely observable only after the choice is made. 
The agent adopts in general a mixed \emph{strategy}, i.e., they make choices according to a probability distribution over the available alternatives. We focus here on a class of mixed strategies, well-known in RUT, namely the \emph{logit model} \cite[\S{}3]{mckelvey1995quantal}: if the agent has a systematic utility function~$v(\cdot,\cdot)$ and the regulator proposes incentive~$y$, they will choose alternative~$i$ with probability:
\begin{equation}
\label{eq:logit}
\mathbb P_i^v(y)
= \frac{
    e^{v(i, y)/\theta} 
   }{
   \sum_{i'\in\mathcal I} e^{v(i', y)/\theta}
   }.
\end{equation}
The probabilistic nature of the choice accounts for the bounded rationality of the agent \cite{wheeler2018bounded}, who does not always make the theoretically optimal decisions, due to limited information. In our case, the agent does not know in advance the actual outcome~$v(i,y)+\epsilon_i$ of their choice, which is revealed only \emph{after} the choice is made. This is because the actual realization of~$\epsilon_i$ is not known before alternative~$i$ is chosen.
Note that, the larger the scale parameter, the greater the randomness in the agent's choice.

The behavior of the agent is determined by their \emph{profile}, capturing inherent agent's preferences or needs impacting their choices. A profile is mapped to a certain systematic utility function. The set of profiles is~$\mathcal K=\{0,1,..,K-1\}$ and their respective utility functions are $u_k(\cdot,\cdot)$.

We allow any finite number $K$ of agent profiles. According to standard RUT (\cite[\S{}3.2.5a]{Cascetta2009}), each alternative~$i$ is described with a (possibly) agent-dependent attribute vector $\mathbf x_{i,k}$. Considering the case in which alternatives represent transport modes, the components of~$\mathbf x_{i,k}$ may represent monetary cost, walking travel time, in-vehicle time, comfort, etc. If alternatives represent different products, $\mathbf x_{i,k}$ may instead represent characteristics such as the price, the reputation of the brand, the weight, etc. Some components are agent-dependent: e.g., the in-vehicle time to reach a certain destination depends on the starting location, e.g., the residence, which varies from one agent to another. This is why~$\mathbf x_{i,k}$ is indexed by~$k$. Each profile is associated with a coefficient vector~$\boldsymbol\alpha_k$, which captures how individuals with profile~$k$ value the various attributes. For instance, older individuals may assign greater importance to walking time compared to younger individuals. The systematic utility is thus $\boldsymbol\alpha_k \cdot \mathbf x_{i,k}$ in the absence of any incentive policy. 

\begin{lemma}
\label{lem:u-profiles}
Assuming utility is quasi-linear with respect to income\footnote{
This is a standard assumption in econometrics (\cite{javaudin2023personalised}), which means that both $\boldsymbol\alpha_k \cdot \mathbf x_{i,k}$ and the incentive~$y$ are expressed in the same unit as income and that $y$ has an additive effect on utility.
} utility function $u_k(\cdot,\cdot)$ can be written as
\begin{align}
\nonumber
&u_k(i, y) =
    \begin{cases}
        a_k+y  & i=1 \\
        b_{k} & i\neq 1
    \end{cases}
;
\begin{array}{l}
 a_k 
 := \boldsymbol\alpha_k \cdot \mathbf x_{1,k};
 \\
 b_k  
 := \theta \ln 
 \sum\limits_{i>0} e^\frac{\boldsymbol\alpha_k\cdot \mathbf x_{i,k}}{\theta}
\end{array}
\\
\label{eq:pk}
&\text{Also, }
p_k(y):=\mathbb P_1^{u_k}(y) 
=
\sigma\left(\frac{a_k- b_k +y}{\theta}\right)
\\
\nonumber
&\text{where }
\sigma(\cdot) \text{ is the sigmoid}
\end{align}
\end{lemma}

\noindent This equation summarizes the fact that, depending on their profiles, some agents might have higher disutility than other when choosing alternative $1$, and we compensate them for this loss, adapting incentive~$y$ to their case. This personalized compensation introduces vertical equity \cite[\S{}3.2]{DanielHo2022}, since it distributes incentives to those who need it the most to change their behavior. Quantity~$a_k-b_k$ interestingly captures how much profile~$k$ is \emph{inherently responsible}, i.e., how much they prefer the more sustainable alternative~$i=1$ over the others, independent of the proposed incentive. We call~$a_k-b_k$ the \emph{responsibility intensity} of that profile. With no lack of generality, we order all profiles accordingly:
\begin{flalign}
\label{eq:ordering}
a_{k-1}-b_{k-1} \ge a_k-b_k, \text{ for } k=1,..,K-1
\end{flalign}

\begin{lemma}
\label{lem:expectations}
Let us consider an agent following a certain profile $k\in\mathcal K$. Then, the expectation of the \emph{externalities} caused by them and of the \emph{incentive collected} are, respectively:
\begin{align}
    \textit{ext}^k(y)
    &=
    p_k(y) \cdot c_1 +
    \left( 1- p_k(y) \right) \cdot c_{\neg 1};
    \nonumber
    \\
    \textit{inc}^k(y)
    &= p_k(y) \cdot y
    \nonumber
    \\
    \text{with } c_{\neg 1} 
    &:= \sum_{i>1} \frac{e^\frac{u_0(i,y)}{\theta} }{ \sum_{i'>1} e^\frac{u_0(i',y)}{\theta} } c_i 
    > c_1
\end{align}
\end{lemma}



The generalized social cost when the agent chooses~$i$ consists of social cost~$c_i$ and the cost of incentive~$y_i$ provided by the regulator, the latter term weighted by a parameter~$\beta>0$, which expresses the relative importance the regulator gives to incentive-related expenditures versus externalities and also normalizes the units of measure. The expected social cost is
\begin{align}
\label{eq:objective_function}
\tilde f^k(y) = 
\textit{ext}^k(y)
+ \beta \cdot \textit{inc}^k
\end{align}

Note that computing the expected value above requires knowing the utility~$\hat v(\cdot,\cdot)$ that the agent is adopting. This is information is not available to the regulator, which can only thus only estimate this expectation, based on estimate~$\hat v(\cdot,\cdot)$, inferred via preference elicitation.




\subsection{Interaction between Regulator and Agent}
\label{sec:interaction}
\vspace{-0.1cm}
The interaction between the regulator and the agent is a sequential game. In each timeslot~$t$:
\begin{enumerate}
    \item \label{it:inference}
    The regulator infers profile~$\hat k^t=\psi(\mathbf h^{t-1})$ of the agent, via a deterministic inference rule~$\psi{\cdot}$ applied on previous observations~$\mathbf h^{t-1}= y^1, i^1, \dots, y^{t-1}, i^{t-1}$, where~$y^s,i^s$ are the incentives proposed by the regulator and the alternative selected by the agent at iteration~$s=1,\dots,t-1$. In the numerical results, inference rule~$\psi(\cdot)$ is the maximum likelihood estimator, details in~\eqref{eq:argmax-logL} of the supplemental material.

    \item \label{it:incentive-calculation}
    The regulator takes its \emph{action}. An action is the incentive~$y^t$ that the regulator proposes to the agent at timeslot~$t$ for choosing alternative~$1$. The regulator proposes incentive~$y^t=\phi(\hat k^t)$, which is associated with inferred profile~$\hat k^t$ via deterministic function~$\phi(\cdot)$. We assume that this function returns the incentive so as to minimize the estimated expectation of the generalized social cost~\eqref{eq:objective_function}. 

    \item \label{it:choice}
    The agent chooses alternative~$i^t\in\mathcal I$, according to ~$\mathbb P^v_i(y^t)$ formulated in ~\eqref{eq:logit}.

    \item  \label{it:transfer}
    The regulator transfers incentive~$y^t$ to the agent, if~$i^t=1$, or no incentive if~$i^t>1$.

    \item \label{it:utility}
    The agent gets utility~$v(i^t,y^t)+e_{i^t}$, where~$e_{i^t}$ is a realization of random noise~$\epsilon_{i^t}$.
\end{enumerate}

The regulator may interact with multiple agents. We focus on one regulator-agent interaction, assuming cross-agent effects do not materially alter individual choices. This standard assumption~\cite{javaudin2023personalised} is appropriate because realistic incentive policies, even when producing non-negligible effects, rarely induce shifts massive enough to substantially change the surrounding choice environment, e.g., affecting congestion.
Since in this case aggregate social cost is additive: the system then decomposes into separate regulator-agent interactions, and our analysis applies to each agent.
Regulator budget constraints could be added in step~\ref{it:incentive-calculation}, which we do not do as it would not add conceptual added value. Moreover, the regulator could solve complex optimization problems when calculating~$y^t$, e.g., to maximize some long-term cumulative objective. However, such complex optimizations would not be implementable by real decision makers: for political acceptance and transparency to citizens, policies must be based on simple, understandable rules, to prevent citizens from wrongly perceiving that incentives are calculated by an arbitrary black box. Our personalized incentives are instead easy to understand by the agents, since they are based on the simple concept that an agent with profile~$k$ deserves a certain incentive~$y_k^\star$, whose value can be transparently announced by the regulator a priori. Therefore,
$
y^t = y_{{\hat k}^t}^\star
$.
Prop.~\ref{eq:y_star} will detail the way in which such a value is calculated.


\section{Analysis}
\subsection{Optimal Incentives}

There are unique incentives that minimize the estimated expectation of the generalized social cost~$\tilde f^{\hat v^t}(y)$:

\begin{proposition} 
\label{prop:y_star}
    There is a unique optimal incentive for inferred profile $u_k(\cdot,\cdot)$:
    \begin{equation}
    \label{eq:y_star}
        y^\star_{u_k}=\max\left(\frac{c_{\neg 1}-c_1}{\beta}-\theta \left(1+W(e^{\frac{a_k-b_k}{\theta}+\frac{c_{\neg 1}-c_1}{\theta \beta}-1}) \right), 0\right) \nonumber
    \end{equation}
    where W(.) is the Lambert's W-function. \emph{Moreover,}
\end{proposition}

\begin{corollary} \label{cor:incentive-comparison}
    Let $y^*_k$ be the optimal incentives as defined in Proposition \ref{prop:y_star}. Then $y^*_k$ is non-increasing with responsibility intensity $a_k-b_k$.
\end{corollary}

This proposition tells us which profiles expect to have higher incentives. Thanks to ordering~\eqref{eq:ordering} and Cor.~\ref{cor:incentive-comparison} we can thus determine the order of proposed incentives.

The following proposition prevents optimal incentives from being infinitely large, which is crucial when evaluating the economic feasibility of the regulator's policy.

\begin{proposition} \label{prop:bound-y-star}
    The optimal incentive is upper-bounded, i.e., $$0 \leq y^\star_{\hat v^t} < \frac{c_{\neg 1}-c_1}{\beta}.$$
\end{proposition}

\noindent As expected, the previous bound grows with the externality~$c_{\neg 1}$ generated when not choosing alternative~$1$ and decreases with the social cost~$c_1$ of alternative~$1$ since the regulator can prevent more social cost by encouraging alternative~$1$ if there is larger difference between the costs caused by alternatives~$1$ and the others. 

\subsection{Agent's Strategic Behavior}
\label{sec:agent-s-strategic-behavior}
\label{sec:manipulation}

We focus here on an agent that is inherently responsible, i.e., their  true profile is the \emph{responsible profile}~$v(\cdot,\cdot)=u_0(\cdot,\cdot)$ (see ordering~\eqref{eq:ordering}) and who has a natural  preference toward sustainable alternatives, i.e., $a_0>b_0$. We study in which case such an agent may be better off manipulating, by adopting (i.e., simulating) another profile. 

We assume that the agent has \emph{bounded rationality}, i.e., they seek to maximize their utility subject to cognitive and computational limitations, which prevents them from adopting arbitrarily complex optimization. This is compliant with human behavior \cite{wheeler2018bounded} and rules out sophisticated Markovian or RL planning strategies that require accounting for long-term rewards. We thus assume that, if the agent decides to be manipulative, they will imitate another profile~$r\in\mathcal K$, via adopting utility~$u_r(\cdot,\cdot)$, i.e., consistently choosing alternative $i^t=1$ or $i^t>1$ according to probability~$p^r(\cdot)$ (see Lemma~\ref{lem:u-profiles}) in all the timeslots.
In the following analysis, among profiles~$\mathcal K$, three deserve particular attention: 
\begin{itemize}
\item $k=0$ is the most responsible profile. We focus on a user whose true profile is~$0$ and study under which conditions they are better off manipulating.
\item $r$ is the profile
adopted by the agent, which determines probability $p^r(\cdot)$ to choose the sustainable alternative. If the agent behaves honestly, $r=0$, otherwise, $r>0$.
\item $\hat k^t$ is the regulator's estimate at time~$t$ of the agent's profile.
\end{itemize}

The true utility collected by the agent when the regulator proposes incentive~$y$ is
\begin{align}
\label{eq:u_0_y}
u_0(y)=
\begin{cases}
a_0+y & \text{if choosing }i=1
\\
b_0 & \text{if choosing }i>1,
\end{cases}
\end{align}
no matter whether the agent is honest ($r=0$) or not ($r>0$).

When a user is manipulative, it incurs a loss of true utility, due to the fact that they do not choose alternatives in accordance with their true profile. On the other end, they have a gain, obtained via shifting the probability distribution of the proposed incentive~$Y_t$ toward a more favorable one. We capture this loss and gain in the definitions of Gain from Manipulation and Cost of Manipulation, which will be given in Lemma~\ref{lem:decomposition}. We need first to introduce some notation.

\begin{definition}
\label{def:notation-of-the-random-variables}
[Notation of the random variables]
Let~$\hat K^t=\psi(\mathbf H^{t-1})$ be the random variable (r.v.) denoting the profile inferred by the regulator at time~$t$. It is the result of applying a deterministic function~$\psi(\cdot)$ (see~\eqref{eq:argmax-logL}) to~$\mathbf H^{t-1}$, which is the r.v. representing the history.  $\mathbf H^{t-1}$ is stochastic since the agent's choices are stochastic (Lemma~\ref{lem:u-profiles}). Let $Y^t=\phi(\hat K^t)$ denote the incentive proposed at time~$t$. It is a r.v., whose possible values are $y_k^\star, k\in\mathcal K$.

Let~$U_0(y)$ denote the random variable representing such a utility. It takes values $a_0+y$ and $b_0$ with probability $p_k(y)$ and $1-p_k(y)$, respectively (Lemma~\ref{lem:u-profiles}). 
The true utility collected by the agent at time~$t$ is thus random variable $U_0(Y^t)$. Its stochasticity comes from the stochasticity of~$Y^t$, as well as the stochasticity of~$U_0(y)$ for any value of~$y$.

For notational compactness, we introduce the following conditional expectation and probability operators
\begin{align}
\label{eq:conditional-expect}
\mathbb E_k[\cdot] = 
\mathbb E[ \cdot | \text{agent adopts profile }k],
\\
\label{eq:conditional-prob}
\mathbb P_k[\cdot] = 
\mathbb P[ \cdot | \text{agent adopts profile }k]
.
\end{align}

Let r.v.~$Y_k^t$ and~$\hat K_k^t$ represent the incentive proposed at time~$t$ and the profile inferred by the regulator, when the agent adopts profile~$k$. Observe that all r.v. $\hat K_k^t,\forall k\in\mathcal K, \forall t$ have the same domain~$\mathcal K$. They only differ in how the probability mass, i.e.,~$P_k[A],\forall \mathcal A\subseteq \mathcal K$~\eqref{eq:conditional-prob},  is distributed in the domain. Similarly, all r.v. $Y_k^t,\forall k\in\mathcal K, \forall t$ have the same domain~$\{y_k^\star\}_{k\in\mathcal K}$.  This notation makes it clear that, by adopting a certain profile~$k\in\mathcal K$, the agent induces a certain probability distribution over the inferred profile, and thus over the incentive proposed.

\end{definition}

\begin{lemma}
\label{lem:decomposition}
Let~$\textit{CM}_r(y)$ be the \emph{Cost of Manipulation} incurred by an agent simulating profile~$r$ when the proposed incentive is $y\in\{y_k^\star\}_{k\in\mathcal K}$:
\begin{align}
\label{eq:CMry}
\textit{CM}_r(y)
:=
\mathbb E_0[U_0(y)]-\mathbb E_r[U_0(y)].
\end{align}

Let~$\textit{CM}_r^t$ be its mean:
\begin{align}
\label{eq:CM}
\textit{CM}_r^t
:=
\mathbb E_0 [\textit{CM}_r (Y^t)]
=
\sum_{k\in\mathcal K}
\begin{array}{r}
\mathbb P_0(Y^t=y_k^\star) 
\\ \cdot
\textit{CM}_r(y_k^\star)
\end{array}
\end{align}

Let~$\textit{GM}_r^t$ be the \emph{Gain from Manipulation} when simulating profile~$r$, defined as

\begin{align}
\label{eq:GM}
\textit{GM}_r^t
:=
\sum_{k\in\mathcal K}
\left(
\begin{array}{r}
\mathbb P_r[Y^t=y_k^\star] 
\\
-
\mathbb P_0[Y^t=y_k^\star]
\end{array}
\right)
\mathbb E_r[U_0(y_k^\star)].
\end{align}

Let~$\Delta U_r^t$ be the \emph{manipulation advantage}, defined as the additional expected true utility collected by the agent at time~$t$ when simulating profile~$r$ over acting honestly. 

Such an advantage can be decomposed in
\begin{align}
\label{eq:decomposition}
\Delta U_r^t:=
\mathbb E_r[U_0(Y^t)]-    
\mathbb E_0[U_0(Y^t)]
=
\textit{GM}_r^t - \textit{CM}_r^t.
\end{align}
\end{lemma}
\begin{proof}
By the law of total expectation, for each $j\in\{0,r\}$,
$
\mathbb E_j[U_0(Y^t)]
=
\sum_{k\in\mathcal K}
\mathbb P_j(Y^t=y_k^\star)\,
\mathbb E_j[U_0(y_k^\star)]
$. Replacing these expressions in~\eqref{eq:decomposition} and adding and subtracting~$\sum_{k\in\mathcal K}\mathbb P_0(Y^t=y_k^\star) \mathbb E_r[U_0(y_k^\star)]$ gives the Lemma.
\end{proof}

We now study how the agent can manipulate the probability distribution of the incentives. Let~$\textit{IE}_r^t$ denote the \emph{intensity of elicitation} of profile~$r$ at time~$t$, i.e., how much the agent can change the probability distribution of the profile inferred by the regulator, by simulating profile~$r>0$:
\begin{align}
\label{eq:IE_r}
\textit{IE}_r^t := D_\text{KL}(\hat K_0^t \| \hat K_r^t),
\end{align}
which is the Kullback-Leibler~(KL) divergence between the profile inferred by the regulator when the agent adopts true profile~$0$ or simulates profile~$r$, respectively. 

Similarly, let~$\textit{IP}_r^t$ denote the \emph{intensity of personalization} of profile~$r$ at time~$t$, i.e., how much the agent can change the probability distribution of the incentives proposed by the regulator, by simulating profile~$r>0$:
\begin{align}
\label{eq:IP_r}
\textit{IP}_r^t := D_\text{KL}(Y_0^t \| Y_r^t)
\le \textit{IE}_r^t,
\end{align}
where the inequality is obtained thanks to
the contraction of divergence induced by function~$\hat Y_k^t=\phi(\hat K_k^t)$ \cite[Theor.~7.4]{PolyanskiyWu2025}.

\begin{theorem}
\label{theor:manipulation-gain-bound}
The gain from manipulation is bounded by
\begin{align}
0
&\le 
\textit{GM}_r^t 
\le
R_r
\sqrt{ \textit{IP}_r^t/2}
\le
R_r
\sqrt{ \textit{IE}_r^t/2}
, 
\\
\label{eq:R} \text{where }
R_r
&:=
\mathbb E_r[U_0(y_{K-1}^\star)] - 
\mathbb E_r[U_0(y_0^\star)]
\end{align}
is the range of expected true utility.
\end{theorem}

\begin{proof}
Thanks to Prop.~\ref{prop:y_star}, $Y_r^t$ First-Order-Stochastically Dominates (FOSD)~$Y_0^t$, which intuitively means that by simulating~$r>0$ the agent shifts the probability mass toward higher values of incentives.
Observe that~$y\mapsto\mathbb E[U_0(y)]$ is increasing. Looking at the sum defining~$\textit{GM}_r^t$~\eqref{eq:GM}, the terms with high values of~$\mathbb E[U_0(y)]$ are multiplied by positive numbers ($\mathbb P_r[\cdot]> \mathbb P_0[\cdot]$, due to FOSD) and only low values of~$\mathbb E[U_0(y)]$ are multiplied by negative numbers. Simple calculation thus yields a positive sum and thus~$\textit{GM}_r^t \ge 0$.

As for the upper bound, we will first bound~$\textit{GM}_r^t$ in terms of total variation~$\delta(Y_0^t,Y_r^t)$ between the probability distributions of~$Y_0^t$ and $Y_r^t$ (def. in~\cite[(7.3)]{PolyanskiyWu2025}). 

Since~$\sum_{k\in\mathcal K}
\left(
\mathbb P_r[Y^t=y_k^\star] - 
\mathbb P_0[Y^t=y_k^\star]
\right)=0$,
we can introduce any constant~$A$ in~ the expression of~$\textit{GM}_r^t$ (Lemma~\ref{lem:decomposition}) without changing its value, as follows:
\begin{align}
\nonumber
\textit{GM}_r^t
=
\sum_{k\in\mathcal K}
\left(
\begin{array}{ll}
\mathbb P_r[Y^t=y_k^\star] 
\\- \mathbb P_0[Y^t=y_k^\star]
\end{array}
\right)
\cdot (\mathbb E_r[U_0(y_k^\star)]-A)
\\
\nonumber
=\frac{R_r}{2}
\sum_{k\in\mathcal K}
\left(
\begin{array}{ll}
\mathbb P_r[Y^t=y_k^\star] 
\\
- \mathbb P_0[Y^t=y_k^\star]
\end{array}
\right)
\cdot \frac{\mathbb E_r[U_0(y_k^\star)]-A}{R_r/2}
\\
\nonumber
=\frac{R_r}{2}
\sum_{k\in\mathcal K}
\left(
\begin{array}{ll}
\mathbb P_r[Y^t=y_k^\star] 
\\
\nonumber
- \mathbb P_0[Y^t=y_k^\star]
\end{array}
\right)
\cdot f(y_k^\star),
\end{align}
where we define function~$f(y_k^\star)=\tfrac{\mathbb E_r[U_0(y_k^\star)]-A}{R_r/2}$. The previous expression is equivalent to
\begin{align}
\textit{GM}_r^t
= \frac{R_r}{2} \left(
\mathbb E_r[f(Y^t)] - \mathbb E_0[f(Y^t)]
\right)
\end{align}

Let us choose in particular
\begin{align}
\nonumber
A&=\left(
\mathbb E_r[U_0(y_{K-1}^\star)]+
\mathbb E_r[U_0(y_0^\star)]
\right)/2
\end{align}
Since $y\mapsto \mathbb E_r[U_0(y)]$ is an increasing function. Also, thanks to Prop.~\ref{prop:y_star}, $y_k^\star$ is non decreasing with~$k$. Since~$A$ is the midpoint between the maximum and minimum expected value,
\begin{align}
|A-\mathbb E_r[U_0(y_k^\star)]|\le R_r/2,\forall k\in\mathcal K.
\end{align}
and thus~$|f(y_k^\star)|\le 1, \forall k\in\mathcal K$. 
We can thus apply the sup-representation of total variation from \citet[(7.18)]{PolyanskiyWu2025} and obtain
\begin{align}
\label{eq:sup-representation}
\textit{GM}_r^t \le R_r\cdot \delta(Y_0^t, Y_r^t).
\end{align}

Thanks to the Pinsker inequality first \cite[(7.27)]{PolyanskiyWu2025}
\begin{align}
\label{eq:pinsker}
\delta&(Y_0^t, Y_r^t) 
\le 
\sqrt{ \tfrac{1}{2} D_\text{KL}(Y_0^t\|Y_r^t)},
\end{align}

Combining~\eqref{eq:sup-representation}, \eqref{eq:pinsker} and~\eqref{eq:IP_r} concludes the proof.
\end{proof}

Theor.~\ref{theor:max-shift} formalizes that the more the intensity of personalization, the more the agent can, at each timeslot, gain from manipulation. Theor.~\ref{theor:max-shift} result adds that, complying with intuition, the more the interactions between the regulator and the agent, the more the latter can manipulate the former. However, Theor.~\ref{theor:max-shift} limits the rate of growth of the potential for manipulation: such a growth can be at most linear.

\begin{theorem}
\label{theor:max-shift}
Manipulation potential~$\textit{IE}_r^t$ is bounded by
\begin{align}
\label{eq:bound-on-Gamma_r}
\textit{IE}_r^t
&\le (t-1) \textit{IE}_r
\\
\label{eq:gammar}
\text{where }
\textit{IE}_r:= 
&\max_{k\in\mathcal K} 
D_\text{KL}\left(
\text{Ber}(p_0(y_k^\star))
\|
\text{Ber}(p_r(y_k^\star))
\right),
\end{align}
and $\text{Ber}(p)$ is a Bernoulli distribution with probability~$p$. 
Moreover, the gain from manipulation is bounded by
\begin{align}
\label{eq:bound-on-GM}
0
&\le 
\textit{GM}_r^t 
\le
R_r
\sqrt{\tfrac{1}{2}(t-1) \textit{IE}_r}.
\end{align}
\end{theorem}

\begin{proof}
The history of proposed incentives and chosen alternatives is~$\mathbf h^{t-1}=y^1,i^1,..,y^{t-1},i^{t-1}$. Therefore, the realization of~$\hat K^t$ is obtained as a deterministic function~$\psi(\mathbf h^{t-1})$. Let~$\mathbf H_k^{t-1}$ be the random variables representing the possible histories when the agent adopts profile~$k$. Then

\begin{align}
\label{eq:polyanskiy}
&\textit{IE}_r^t = D_\text{KL}(\hat K_0^t\| \hat K_r^t)
\le D_\text{KL}(\mathbf H_0^{t-1}\|\mathbf H_r^{t-1})
\\
\label{eq:chain}
&\le 
\left[
\sum_{s=1}^{t-1} 
\sum_{\mathbf h\in\mathcal H^{s-1}}
\begin{array}{l}
\mathbb P_0[\mathbf H^{s-1}=\mathbf h]\cdot 
\\
D_\text{KL}
\left( 
\begin{array}{l}
\mathbb P_0 
(
I^s | \mathbf H^{s-1}=\mathbf h
)
\\
\|
\mathbb P_r 
(
I^s | \mathbf H^{s-1}=\mathbf h
)
\end{array}
\right)
\end{array}
\right]
\\
\label{eq:bernoulli}
&=
\left[ 
\sum_{s=1}^{t-1} 
\sum_{\mathbf h\in\mathcal H^{s-1}}
\begin{array}{l}
\mathbb P_0[\mathbf H^{s-1}=\mathbf h]\cdot 
\\
D_\text{KL}
\left( 
\begin{array}{l}
\text{Ber}(p_0(y_{\psi(\mathbf h^{s-1})}^\star))
\\
\|
\text{Ber}(p_r(y_{\psi(\mathbf h^{s-1})}^\star))
\end{array}
\right)
\end{array}
\right]
\\
\label{eq:last-ineq}
&\le^\eqref{eq:gammar}
(t-1) \textit{IE}_r
,
\end{align}
Inequality~\eqref{eq:polyanskiy} simply exploits contraction of divergence induced by applying a function \cite[Theor.~7.4]{PolyanskiyWu2025}. Eq.~\eqref{eq:chain} is the chain rule obtained as in~\cite[(9)]{Kleinberg2023} and $\mathcal H^{s-1}$ is the set of all possible histories up to time~$s-1$.
Finally, Eq~\eqref{eq:bound-on-GM} is the combination of Theor.~\ref{theor:manipulation-gain-bound} and~\eqref{eq:bound-on-Gamma_r}.
\end{proof}

Is it thus possible to understand up to which timeslot manipulation is deterred, as formalized in the following corollary.

\begin{corollary}
\label{lem:timeslot}
The manipulation of the regulator via imitating profile~$r$ is deterred in any timeslot~$t$ s.t.
\begin{align}
\label{eq:t_deterrent}
t&<1+\frac{2}{\textit{IE}_r}
\left(
\frac{\textit{CM}_r^\text{min}}{R_r}
\right)^2
\\
\text{where }
& 
\label{CMrmin}
\textit{CM}_r^\text{min} = 
\min_{k\in\mathcal K}
\textit{CM}_r(y_k^\star).
\end{align}

\end{corollary}

\begin{proof}
Lemma~\ref{lem:decomposition} implies that manipulation is beneficial only if~$\textit{GM}_r^t-\textit{CM}_r^t>0$:
\begin{align}
0<\textit{GM}_r^t-\textit{CM}_r^t 
\le
R_r
\sqrt{\tfrac{1}{2}(t-1) \textit{IE}_r}, 
-
\textit{CM}_r^\text{min}
\end{align}
The last inequality is obtained by applying~\eqref{eq:bound-on-GM} and~\eqref{CMrmin}. Therefore, \eqref{eq:t_deterrent} ensures manipulation is deterred.
\end{proof}

\section{Impact of Preference Elicitation} 
\label{section:impact-PE}
\label{sec:impact-of-preference-elicitation}

We now examine how preference elicitation limits the agent’s ability to manipulate the regulator. It not only makes it harder for the agent to convince the regulator that their adopted profile is~$r>0$, but also substantially reduces the expected utility gained from manipulation.

We compare two preference-elicitation models. The first, \emph{Instantaneous Preference Elicitation} (\emph{InstPE}), is implicitly assumed in previous artificial intelligence works~\cite{dawkins2021limits,birmpas2020optimally,gan2019manipulating,gan2019imitative,nguyen2019imitative}. Under \emph{InstPE}, the regulator immediately accepts the agent’s simulated utility profile, so that $\hat k^t=r$ for all~$t$. This corresponds to strategic misreporting and allows an agent behaving according to profile~$r$ to immediately maximize their expected utility. 

We instead introduce \emph{Realistic Preference Elicitation} (\emph{RealPE}). Here, the regulator observes the agent’s choices over time and updates its belief~$\psi$ about the agent’s utility function at each timeslot~$t$. The agent must therefore establish the simulated profile through repeated interaction and may meanwhile be classified under another profile. This repeated simulation is costly because it may require choosing alternatives that the agent does not prefer.
To distinguish the two settings, we denote expectations by $\mathbb E^\textit{InstPE}$ and $\mathbb E^\textit{RealPE}$. Under \emph{InstPE}, the agent can directly compare the expected utilities of honest and manipulative behavior.

\begin{observation}
\label{obs:instPE}
In \textit{InstPE} (i)~the regulator proposes incentive~$y_k^\star$ if the agent declares profile~$k\in\mathcal K$ and (ii)~the agent always behaves according to their true profile~$0$, no matter the declared profile~$k$. Therefore, the operators~\eqref{eq:conditional-prob}-\eqref{eq:conditional-expect} are such that~$\mathbb P_k[\cdot] = \mathbb P_0[\cdot]$ and $\mathbb E_k[\cdot] = \mathbb E_0[\cdot]$.
\end{observation}

The following Lemma shows that the cost of manipulation is null under $\textit{InstPE}$ and the gain from manipulation is directly the difference between the expected true utilities, under incentive~$y_r^*$ of simulated profile~$r$ and under incentive~$y_0^\star$ associated to the true profile~$0$.

\begin{lemma}
\label{lem:InstPE}
\begin{align}
\nonumber
\textit{GM}_r^{t,\textit{InstPE}} = 
\mathbb E_0[U_0(y_r^\star)] -
\mathbb E_0[U_0(y_0^\star)],
&& \textit{CM}_r^t=0,
&& \forall t
\end{align}
\end{lemma}

\section{Case Study: Travel-Mode Incentives} \label{section:num-results}

We here validate the theoretical findings and show the deterring impact of Realistic Preference Elicitation (\emph{RealPE}) on manipulative users, compared to Instantaneous Preference Elicitation (\emph{InstPE}). We run a comprehensive Monte Carlo simulation campaign calibrated on transport mode choices in London (open dataset of \citet{LondonDataset}), which reports choices of travelers (agents), with $i=1$~(bike), $i=2$~(public transport (PT)), $i=3$~(car). We calibrated the parameters related to $3$ profiles~$k=0,1,2$ and the respective behavioral parameters, and ordered them according to the responsibility intensity. Without imposing it, the calibration returned profiles with preference toward biking, PT, and car. We then run Monte Carlo simulation over multiple seeds, to implement the stochastic choice of the agent, the profile inference (preference elicitation) of the regulator and the calculation of the proposed incentive. The important parameters assigned based on previous experimented incentive policies and CO\textsuperscript{2} emission evaluations, as well as values obtained via calibration and~\eqref{eq:y_star} are
\begin{scriptsize}
 $\theta=3.7$, $\{a_k\}_k=\{6.1, -6.8, -6.6\}$, $\{b_k\}_k=\{5.8, -11.7, -11.9\}$, $\{c_{\-neg 1,k}\}_k=\{2.8, 2.2, 2.6\}$, $\beta=0.6$, $\{y_k^\star\}_k=\{0, 0.1, 0.9\} (\pounds)$
\end{scriptsize}.

A detailed explanation of the dataset, the hyperparameters, the calibration method and the parameters resulting from calibration \S{}\ref{sec:dataset-and-calibration} of the of the Suppl.Material. A Python notebook that reproduces dataset downloading and preprocessing, calibration and simulation is provided in the Suppl.Material. It runs entirely in less than 2 minutes in Google Colab.

\begin{figure}
    \centering
    \includegraphics[width=1\linewidth]{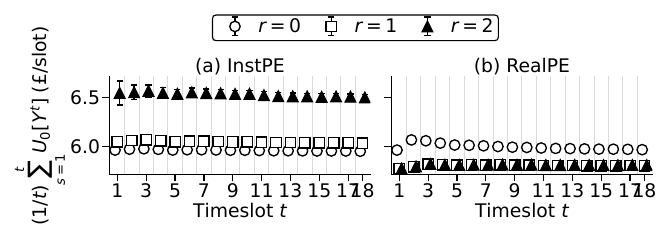}
    \caption{Time-averaged true agent utility. Averages across all the random realizations (corresponding each to a random seed) and 95\% confidence intervals are represented}
    \label{fig:fig_time_averaged_utility}
\end{figure}

\begin{figure}
    \centering
    \begin{subfigure}[t]{0.48\linewidth}
        \centering
        \includegraphics[width=\linewidth]{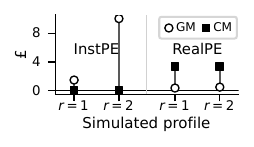}
        \caption{}
        \label{fig:manipulation}
    \end{subfigure}
    \hfill
    \begin{subfigure}[t]{0.48\linewidth}
        \centering
        \includegraphics[width=\linewidth]{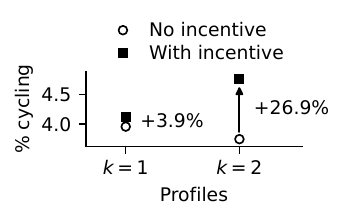}
        \caption{}
        \label{fig:behavioral-shift}
    \end{subfigure}

    \caption{(a) Cost of and gain from manipulation ($\textit{CM},\textit{GM}$); 
    (b) change in the percentage of cycling trips}
\end{figure}

Fig.~\ref{fig:manipulation} shows that Cost of Manipulation is higher than the Gain under \textit{RealPE}. Due to~\eqref{eq:decomposition}, time-averaged true utility~$(1/t)\sum_{s=0}^t U_0(Y^t)$ is higher when acting honestly under \textit{RealPE} (Fig.~\ref{fig:fig_time_averaged_utility}). Manipulation is thus deterred. Fig.~\ref{fig:fig_time_averaged_utility}-\ref{fig:manipulation} also shows that the situation would be inverted under the over-pessimistic \textit{InstPE} setting. Therefore, under realistic conditions, personalized incentive policies can be effective in fostering sustainable behavior, as shown in Fig.~\ref{fig:behavioral-shift} where we computed the induced shift in bike usage.


\section{Conclusion}
By explicitly modeling the regulator’s learning process in personalized incentive policies, we establish an upper bound on the gain achievable from manipulative agents. Interestingly, such a bound tightens as the regulator’s ability to distinguish profiles decreases. We characterize the resulting deterrence effect and show, through simulations calibrated on observed travel choices, how it improves policy effectiveness. Our analysis can serve as a pre-deployment test, helping regulators assess the policy's robustness to manipulation.

\newpage
\bibliography{aaai2027}

\appendix

\renewcommand{\thesubsection}{%
  \thesection\arabic{subsection}%
}
\renewcommand{\thesubsubsection}{%
  \thesubsection.\arabic{subsubsection}%
}

\begin{table*}[!h]
\section*{SUPPLEMENTAL MATERIAL:APPENDIX}

\section{Table of notation (2 pages)}

\centering
\caption{Notation used throughout the paper and the place where each symbol is first introduced.}
\label{tab:notation}
\footnotesize
\setlength{\tabcolsep}{4pt}
\renewcommand{\arraystretch}{1.06}
\begin{tabularx}{\textwidth}{
@{}
>{\raggedright\arraybackslash}p{0.18\textwidth}
>{\raggedright\arraybackslash}X
>{\raggedright\arraybackslash}p{0.25\textwidth}
@{}}
\toprule
\textbf{Symbol} & \textbf{Meaning} & \textbf{First introduced} \\
\midrule
\multicolumn{3}{@{}l}{\textit{Alternatives, profiles, and utilities}}\\[2pt]

$\mathcal I=\{1,\ldots,I\}$, $i$ &
Set of alternatives, its cardinality, and an alternative index; alternative $1$ is the most sustainable one. &
Sect.~\ref{sec:model}, first paragraph \\

$c_i$ &
Externality generated by choosing alternative $i$. &
Sect.~\ref{sec:model}, first paragraph \\

$y\geq0$ &
Incentive offered for choosing alternative $1$. &
Sect.~\ref{sec:model}, first paragraph \\

$v(i,y)$ &
Systematic utility from alternative $i$ under incentive $y$. &
Sect.~\ref{sec:model}, before Eq.~\eqref{eq:logit} \\

$\epsilon_i$, $\theta$ &
Random utility term for alternative $i$ and its Gumbel scale parameter. &
Sect.~\ref{sec:model}, before Eq.~\eqref{eq:logit} \\

$\mathbb P_i^v(y)$ &
Probability that an agent with systematic utility $v$ chooses alternative $i$ under incentive $y$. &
Eq.~\eqref{eq:logit} \\

$\mathcal K=\{0,\ldots,K-1\}$, $k$ &
Set of profiles, its cardinality, and a profile index. &
Sect.~\ref{sec:model}, before Lemma~\ref{lem:u-profiles} \\

$u_k(i,y)$ &
Systematic utility associated with profile $k$. &
Sect.~\ref{sec:model}, before Lemma~\ref{lem:u-profiles} \\

$\mathbf x_{i,k}$, $\boldsymbol\alpha_k$ &
Attribute vector of alternative $i$ for profile $k$, and the corresponding coefficient vector. &
Sect.~\ref{sec:model}, before Lemma~\ref{lem:u-profiles} \\

$a_k$, $b_k$ &
Systematic utility of alternative $1$ without incentives, and aggregate utility of the remaining alternatives, for profile $k$. &
Lemma~\ref{lem:u-profiles} \\

$p_k(y)$, $\sigma(\cdot)$ &
Probability that profile $k$ chooses alternative $1$, and the sigmoid function. &
Eq.~\eqref{eq:pk} \\

$a_k-b_k$ &
Responsibility intensity of profile $k$. &
Immediately before Eq.~\eqref{eq:ordering} \\

$\textit{ext}^k(y)$, $\textit{inc}^k(y)$ &
Expected externality and expected collected incentive when the agent adopts profile $k$. &
Lemma~\ref{lem:expectations} \\

$c_{\neg 1}$ &
Expected externality conditional on not choosing alternative $1$. &
Lemma~\ref{lem:expectations} \\

$\beta$ &
Weight assigned to incentive expenditures relative to externalities. &
Eq.~\eqref{eq:objective_function} \\

$\widetilde f^k(y)$ &
Expected generalized social cost under profile $k$ and incentive $y$. &
Eq.~\eqref{eq:objective_function} \\

$\widehat v(\cdot,\cdot)$, $\widehat v^t(\cdot,\cdot)$ &
Regulator's estimate of the utility profile adopted by the agent. &
After Eq.~\eqref{eq:objective_function} \\

\addlinespace[3pt]
\multicolumn{3}{@{}l}{\textit{Sequential interaction and optimal incentives}}\\[2pt]

$t$, $s$ &
Current timeslot and a generic preceding-timeslot index. &
Step~\ref{it:inference} of Sect.~\ref{sec:interaction} \\

$\mathbf h^{t-1}$ &
Realized history of incentives and choices observed before timeslot $t$. &
Step~\ref{it:inference} of Sect.~\ref{sec:interaction} \\

$\widehat k^t$ &
Profile inferred by the regulator at timeslot $t$. &
Step~\ref{it:inference} of Sect.~\ref{sec:interaction} \\

$\psi(\cdot)$ &
Deterministic profile-inference rule. &
Step~\ref{it:inference} of Sect.~\ref{sec:interaction} \\

$y^t$, $i^t$ &
Incentive proposed and alternative chosen at timeslot $t$. &
Steps~\ref{it:incentive-calculation}--\ref{it:choice} of Sect.~\ref{sec:interaction} \\

$\phi(\cdot)$ &
Deterministic mapping from an inferred profile to an incentive. &
Step~\ref{it:incentive-calculation} of Sect.~\ref{sec:interaction} \\

$y_k^\star$ &
Optimal incentive associated with inferred profile $k$. &
Prop.~\ref{prop:y_star} \\

\bottomrule
\end{tabularx}
\end{table*}

\begin{table*}[!t]
\ContinuedFloat
\centering
\caption{Notation used throughout the paper (continued).}
\footnotesize
\setlength{\tabcolsep}{4pt}
\renewcommand{\arraystretch}{1.06}
\begin{tabularx}{\textwidth}{
@{}
>{\raggedright\arraybackslash}p{0.18\textwidth}
>{\raggedright\arraybackslash}X
>{\raggedright\arraybackslash}p{0.25\textwidth}
@{}}
\toprule
\textbf{Symbol} & \textbf{Meaning} & \textbf{First introduced} \\
\midrule
\multicolumn{3}{@{}l}{\textit{Manipulation analysis}}\\[2pt]

$0$, $r$ &
Agent's true profile and the profile imitated by a manipulative agent. &
Sect.~\ref{sec:manipulation} \\

$u_0(y)$ &
True realized utility under incentive $y$, expressed according to the selected alternative. &
\eqref{eq:u_0_y}\\

$\mathbf H^{t-1}$, $\widehat K^t$, $Y^t$ &
Random history, inferred profile, and proposed incentive at timeslot $t$. &
Def.~\ref{def:notation-of-the-random-variables} \\

$U_0(y)$, $U_0(Y^t)$ &
Random true utility under a fixed incentive $y$, and under the random incentive proposed at timeslot $t$. &
Def.~\ref{def:notation-of-the-random-variables} \\

$\mathbb E_k[\cdot]$, $\mathbb P_k[\cdot]$ &
Expectation and probability conditional on the agent adopting profile $k$. &
Eqs.~\eqref{eq:conditional-expect}--\eqref{eq:conditional-prob} \\

$Y_k^t$, $\widehat K_k^t$ &
Random incentive and inferred profile conditional on the agent adopting profile $k$. &
After Eqs.\eqref{eq:conditional-expect}-\eqref{eq:conditional-prob} \\

$\textit{CM}_r(y)$, $\textit{CM}_r^t$ &
Cost of Manipulation under incentive $y$, and its mean at timeslot $t$. &
Eqs.~\eqref{eq:CMry}--\eqref{eq:CM} \\

$\textit{GM}_r^t$ &
Gain from Manipulation at timeslot $t$ when imitating profile $r$. &
Eq.~\eqref{eq:GM} \\

$\Delta U_r^t$ &
Manipulation advantage: additional expected true utility obtained by imitating profile $r$. &
Eq.~\eqref{eq:decomposition} \\

$D_{\mathrm{KL}}(\cdot\Vert\cdot)$ &
Kullback--Leibler divergence. &
Eq.~\eqref{eq:IE_r} \\

$\textit{IE}_r^t$ &
Intensity of elicitation: distinguishability between the inferred-profile distributions induced by profiles $0$ and $r$. &
Eq.~\eqref{eq:IE_r} \\

$\textit{IP}_r^t$ &
Intensity of personalization: distinguishability between the corresponding incentive distributions. &
Eq.~\eqref{eq:IP_r} \\

$R_r$ &
Range of expected true utility over the incentives the regulator may propose. &
Theor.~\ref{theor:manipulation-gain-bound}, Eq.~\eqref{eq:R} \\

$\delta(\cdot,\cdot)$ &
Total-variation distance between probability distributions. &
Proof of Theor.~\ref{theor:manipulation-gain-bound} \\

$\textit{IE}_r$ &
Maximum one-step KL divergence between the choices induced by profiles $0$ and $r$. &
Eq.~\eqref{eq:gammar} \\

$\operatorname{Ber}(p)$ &
Bernoulli distribution with success probability $p$. &
Theor.~\ref{theor:max-shift}, Eq.~\eqref{eq:gammar} \\

$\mathbf H_k^{t-1}$, $\mathcal H^{s-1}$, $I^s$ &
Random history conditional on profile $k$, set of possible histories, and random alternative at timeslot $s$. &
Proof of Theor.~\ref{theor:max-shift} \\

$\textit{CM}_r^{\min}$ &
Minimum Cost of Manipulation over all incentives that may be proposed. &
Eq.~\eqref{CMrmin} \\

\addlinespace[3pt]
\multicolumn{3}{@{}l}{\textit{Preference-elicitation models}}\\[2pt]

$\textit{InstPE}$, $\textit{RealPE}$ &
Instantaneous and Realistic Preference Elicitation. &
Sect.~\ref{section:impact-PE} \\

$\mathbb E^{\textit{InstPE}}[\cdot]$, $\mathbb E^{\textit{RealPE}}[\cdot]$ &
Expectation under Instantaneous and Realistic Preference Elicitation. &
Sect.~\ref{section:impact-PE} \\

$Z^t:=\mathbf{1}\{I^t=1\}$, $z^t$ &
Binary random variable and its realization indicating whether the
sustainable alternative~$1$ is chosen. &
Appendix~\ref{section:preferenc-elicitation} \\

$L_k^t(\mathbf h^{t-1})$ &
Likelihood of observing history $\mathbf h^{t-1}$ under profile $k$. &
Eq.~\eqref{eq:likelihood} \\

$\ln L_k^t(\mathbf h^{t-1})$ &
Log-likelihood of observing history $\mathbf h^{t-1}$ under profile $k$. &
Eq.~\eqref{eq:loglikelihood} \\

\bottomrule
\end{tabularx}

\end{table*}

\clearpage

\section{Proofs} \label{proofs}
\subsubsection*{Proof of Lemma \ref{lem:u-profiles}}
\begin{lemma*}
Assuming utility is quasi-linear with respect to income,\footnote{
This is a standard assumption in econometrics (\cite{javaudin2023personalised}), which means that both $\boldsymbol\alpha_k \cdot \mathbf x_i$ and the incentive~$y$ are expressed in the same unit as income and that $y$ has an additive effect on utility.
}
utility function $u_k(\cdot,\cdot)$ can be written as
\begin{align}
\label{eq:u-profiles}
&u_k(i, y) =
    \begin{cases}
        a_k+y  & i=1 \\
        b_{k} & i\neq 1
    \end{cases}
;
\begin{array}{l}
 a_k 
 := \boldsymbol\alpha_k \cdot \mathbf x_{1,k};
 \\
 b_k  
 := \theta \ln 
 \sum\limits_{i>1} e^\frac{\boldsymbol\alpha_k\cdot \mathbf x_i}{\theta}
\end{array}
\\
&\text{Also, }
p_k(y):=\mathbb P_1^{u_k}(y) 
=
\sigma\left(\frac{a_k- b_k +y}{\theta}\right)
\\
\nonumber
&\text{where }
\sigma(\cdot) \text{ is the sigmoid}
\end{align}
\end{lemma*}
\begin{proof}
The quasi-linearity of the utility with respect to income implies that the systematic utility associated with alternative~$1$ is $u_k(1,y)=\boldsymbol\alpha_k \cdot \mathbf x_i+y$, and thus $u_k(1,y)=a_k+y$.
If the agent does not choose alternative~$1$, they will choose the best among the remaining alternatives~$i>1$. As explained in \S{}\ref{sec:model}, the random utility of the agent associated with alternative~$i$ is~$u_k(i,y)+\epsilon_i$, where~$\epsilon_i$ is a Gumbel distribution with scale parameter\footnote{The scale parameter is also called ``noise parameter''.}~$\theta$. As suggested by~\cite[\S{}1.8]{Strzalecki2025}, it is reasonable to assume that, when choosing among alternatives~$i>1$, the agent chooses the best among them. The expected utility in this case is
\begin{align}
b_k 
&:=
\mathbb E \left[u_k(i,y) | i > 1 \right]
=
\mathbb E \left[ \max_{i>1} [u_k(i,y)+\epsilon_i ] \right]
\\
&=
\theta \ln \sum_{i>1} e^\frac{v(i,y)}{\theta}
\end{align}
where the last equality is obtained via~\cite[Prop.~3.6]{Strzalecki2025}. Observe that we do not need to carry the specification of~$y$ in symbol~$b_k$, because no incentive is given to any alternative~$i>1$.

By simple calculus $e^\frac{b_k}{\theta}=\sum_{i>1} e^\frac{v(i,y)}{\theta}$. Therefore, probability~\eqref{eq:logit} can be written as:
\begin{align}
\mathbb P_1^{u_k}
&=
\frac{e^\frac{a_k+y}{\theta}}{e^\frac{a_k+y}{\theta} + \sum_{i>1} e^\frac{v(i,y)}{\theta}}
=
\frac{1}{1+e^\frac{b_k}{\theta}/e^\frac{a_k+y}{\theta}}
\\
&=
\frac{1}{1+e^{ -\frac{a_k+y-b_k}{\theta} } }
\end{align}
which is the sigmoid calculated in $\frac{a_k+y-b_k}{\theta}$.
\end{proof}

\subsubsection*{Proof of Lemma~\ref{lem:expectations}}
\begin{lemma*}
Let us consider an agent following a certain profile $k\in\mathcal K$. Then, the expectation of the \emph{externalities} caused by them and of the \emph{incentive collected} are, respectively:
\begin{align}
    \textit{ext}^k(y)
    &=
    p_k(y) \cdot c_1 +
    \left( 1- p_k(y) \right) \cdot c_{\neg 1,k};
    \nonumber
    \\
    \textit{inc}^k(y)
    &= p_k(y) \cdot y
    \nonumber
    \\
    \text{with } c_{\neg 1,k} 
    &:= \sum_{i>1} \frac{e^\frac{u_k(i,y)}{\theta} }{ \sum_{i'>1} e^\frac{u_k(i',y)}{\theta} } c_i 
    > c_1
\end{align}
\end{lemma*}

\begin{proof}
The result is obtained by applying the law of total expectation:
\begin{align}
\nonumber
\textit{ext}^k(y)
&=
\mathbb P_1^{v_k}(y) \mathbb E \big(\textit{ext}(y)| \text{ agent adopts }k, i=1 \big) 
\\ \nonumber
&+
\left( 1- \mathbb P_1^{u_k}(y) \right) 
\mathbb E \big(\textit{ext}(y)| \text{ agent adopts }k, i\neq 1 \big) 
\\ \nonumber
&= 
p_k(y)\cdot c_1
\\
&+
\left( 1- p_k(y) \right) \cdot
\sum_{i>1} \frac{e^\frac{u_k(i,y)}{\theta} }{ \sum_{i'>1 e^\frac{u_k(i',y)}{\theta}} } c_i
\end{align}

The conditional distribution among the alternatives $i>1$ does not
depend on $y$, because the incentive modifies only the utility of
alternative~$1$.

The expected incentive collected is computed similarly.

\end{proof}

\begin{lemma*}
    The estimated expectation of the generalized social cost $\tilde f^k(y):[0,+\infty[ \to \mathbb R$ has a unique minimum. Let~$\hat y_k$ be the root of
    \begin{align}  
    \label{eq:g_root}
    g^k(y) &:= y + \theta + \theta e^{(y+a_k - b_k)/\theta}-\frac{c_{\neg 1,k} - c_1}{\beta}
    \end{align}
    Then,
    \begin{align}
    \text{if } &\frac{c_{\neg 1,k} -c_1}{\beta} > \theta(1+e^{(a_k - b_k)/\theta})
    \label{eq:cond-y-pos}
    \\
    & \text{the min of }\tilde f^k:[0,+\infty[ \to \mathbb R\text{ is at }\hat y_k>0
    \\
    \text{else} & \text{ such a minimum is at }y=0
    \label{eq:minimum-at-zero}
    \end{align}
\end{lemma*}
\begin{proof}
    The estimated expectation of the generalized social cost, which the regulator aims to minimize, is given by \eqref{eq:objective_function}. By Lemmas~\ref{lem:expectations} and~\ref{lem:u-profiles} this is 
    \begin{align}
    \tilde f^k(y)
    &=
    \frac{
    e^{(a_k+y)/\theta}
    }{
    e^{(a_k+y)/\theta}+e^{b_k/\theta}
    }
    \left(c_1+\beta y\right)
    \\&+
    \frac{
    e^{b_k/\theta}
    }{
    e^{(a_k+y)/\theta}+e^{b_k/\theta}
    }
    c_{\neg1,k}.
    \label{eq:expanded-profile-objective}
    \end{align}
    By taking its derivative, we get 
    \begin{align}
    \nonumber
    &   
    (\tilde f^k(y))'
    \\
    \nonumber
    &
    = 
    \frac{
    \beta ye^{(y+a_k + b_k)/\theta} + \beta\theta e^{(y+a_k + b_k)/\theta} + \beta\theta e^{2(y+a_k)/\theta}
    }{
    \theta (e^{(a_k+y)/\theta} + 
        e^{b_k/\theta})^2
    } 
    \\ 
    & + \frac{c_1 e^{(y+a_k + b_k)/\theta}-c_{\neg 1,k} e^{(y+a_k + b_k)/\theta}}{\theta (e^{(a_k+y)/\theta} + e^{b_k/\theta})^2}
    \end{align}

    To find the extrema we take $(\tilde f^k(y))'=0$ and, thus, we obtain
    \begin{align*}
         \beta ye^{(y+a_k + b_k)/\theta} + \beta\theta e^{(y+a_k + b_k)/\theta} + \beta\theta e^{2(y+a_k - b_k)/\theta} \\ -(c_{\neg 1,k}-c_1)e^{(y+a_k - b_k)/\theta} = 0 \nonumber
    \end{align*}
    Dividing by $\beta e^{(y+a_k+b_k)/\theta}$, we obtain
    \begin{equation}
        \iff g^k(y) := y + \theta + \theta e^{(y+a_k - b_k)/\theta}-\frac{c_{\neg 1,k} - c_1}{\beta} = 0
        \label{eq:def-of-hat-g}
    \end{equation}

    \noindent So the extremum of~$\tilde f^k(y):\mathbb R \to \mathbb R$ is at the root of~$g^k(y)$, which we will refer to as~$\hat y_k$. 

    Note that~$g^k(y)$ is an increasing function, which implies that its root~$\hat y_k$ is unique.

    Moreover, observe that~$(\tilde f^k)'$ is equal to~$g^k$ multiplied by a positive quantity, namely $(\beta e^{(y+a_k+b_k)/\theta})/(\theta (e^{(a_k+y)/\theta} + e^{b_k/\theta})^2)$. Therefore,~$g^k(y)>0$ implies~$(\tilde f^k(y))'>0$, i.e.,~$\tilde f^k(y)$ increasing, and~$g^k(y)<0$ implies~$(\tilde f^k(y))'<0$, i.e.,~$\tilde f^k(y)$ decreasing. 

    For~$y<\hat y_k$, function~$\tilde f^k(y)$ is decreasing, since~$y<\hat y_k$ implies~$g^k(y)<g^k(\hat y_k)$, i.e.~$g^k(y)<0$. For~$y>\hat y_k$, function~$\tilde f^k(y)$ is increasing, since~$y>\hat y_k$ implies ~$g^k(y)>g^k(\hat y_k)$, i.e.~$g^k(y)>0$. The root of $g^k$ is, thus, the minimum of $\tilde f^k(y)$ on $\mathbb R \to \mathbb R$.

    Suppose that the root~$\hat y_k$ of~$g^k$ is non-positive. In this case, in the interval~$[0,+\infty[$, function~$\tilde f^k(y)$ is always increasing. Its minimum is thus at~$0$. Otherwise, the minimum of~$\tilde f^k(y):[0,+\infty[\rightarrow \mathbb R$ is the root~$\hat y_k$ of~$g^k$.

    Condition~\eqref{eq:cond-y-pos} determines the sign of the root~$\hat y_k$. Indeed

    \begin{align}
    \text{If }\frac{c_{\neg 1,k}-c_1}{\beta} > \theta(1+e^{(a_k - b_k)/\theta}),
    \\
    \text{ then }
    0 > \theta(1+e^{(a_k - b_k)/\theta}) - \frac{c_{\neg 1,k}-c_1}{\beta}
    \\
    \Rightarrow
    g^k(\hat y_k) > g^k(0)
    \label{eq:inequality-hat-g}
    \end{align}
To obtain the last implication, we simply use the definition of~$g^k(y)$~\eqref{eq:def-of-hat-g} (replacing~$y=0$) and the fact that it is zero at its root, i.e., $g^k(\hat y_k)=0$.

Since~$g^k(y)$ is increasing, equation~\eqref{eq:inequality-hat-g} implies that~$\hat y_k\ge 0$, and so the minimum of~$\tilde f^k:[0,+\infty[\rightarrow \mathbb R$ corresponds to~$\hat y_k$. On the contrary, if condition~\eqref{eq:cond-y-pos} does not hold, the root~$\hat y_k<0$ and the minimum of~$\tilde f^k:[0,+\infty[\rightarrow \mathbb R$ is~$0$.
\end{proof}

\subsubsection*{Proofs of Proposition \ref{prop:y_star} and Corollary \ref{cor:incentive-comparison}}
\begin{prop*} \textbf{\ref{prop:y_star}.}
    There is a unique optimal incentive for profile $k\in\mathcal K$:
    \begin{equation}
        y_k^\star=\max\left(\frac{c_{\neg 1,k}-c_1}{\beta}-\theta \left(1+W(e^{\frac{a_k-b_k}{\theta}+\frac{c_{\neg 1,k}-c_1}{\theta \beta}-1}) \right), 0\right)
    \end{equation}
    where W(.) is the Lambert's W-function.
\end{prop*}
\begin{proof}
    The estimated expectation of the generalized social cost, that the regulator aims to minimize, is given by \eqref{eq:objective_function}. By \eqref{eq:pk} this is  
    \begin{align} 
    \nonumber
    \tilde f^k(y) 
    &= \frac{e^{(a_k +y)/\theta}} {e^{(a_k +y)/\theta}+e^{b_k/\theta}} 
    (c_1+ \beta y) 
    \\ &
    + \frac{e^{b_k/\theta}} {e^{(a_k +y)/\theta}+e^{b_k/\theta}} c_{\neg 1,k}
    \nonumber
    \end{align}
    
    Let $y_k^\star=0$ be point at which~$\tilde f^k(y):[0,+\infty[ \rightarrow \mathbb R$ has its minimum (which is unique - see~\eqref{eq:g_root}).

    Let us start by assuming that~\eqref{eq:cond-y-pos} holds. In this case, $y_k^\star$ is the the root of $g^k(y)$ and is positive (see \eqref{eq:g_root}). Therefore
    \begin{align*}
    y_k^\star + \theta + \theta e^{(a_k - b_k+y_k^\star)/\theta}-\frac{c_{\neg 1,k}-c_1}{\beta} = 0
    \\
        \Rightarrow \frac{y_k^\star}{\theta} + 1 + e^{(a_k - b_k+y_k^\star)/\theta} -\frac{c_{\neg 1,k}-c_1}{\theta\beta} = 0 
    \\
        \Rightarrow e^{(a_k - b_k+y_k^\star)/\theta} = \frac{c_{\neg 1,k}-c_1}{\theta\beta} -\frac{y_k^\star}{\theta} - 1
    \end{align*}
    
    We substitute $u=\frac{c_{\neg 1,k}-c_1}{\theta\beta} -\frac{y_k^\star}{\theta} - 1$ and, by inserting $u$ and $\frac{y_k^\star}{\theta} = \frac{c_{\neg 1,k}-c_1}{\theta\beta} -1 -u$ into the last equation, we  get
    \begin{align*}
       e^{\frac{a_k - b_k}{\theta}+\frac{c_{\neg 1,k}-c_1}{\theta\beta} -1 -u} = u
    \\
        \Rightarrow e^{\frac{a_k - b_k}{\theta}+\frac{c_{\neg 1,k}-c_1}{\theta\beta} -1 } = ue^{u}
    \end{align*}
    
    Now we can apply the Lambert W-function on both sides
    \begin{align*}
      W(e^{\frac{a_k - b_k}{\theta}+\frac{c_{\neg 1,k}-c_1}{\theta\beta} -1 }) = W(ue^{u}) = u
    \end{align*}
    By reinserting $u=\frac{c_{\neg 1,k}-c_1}{\theta\beta} -\frac{y_k^\star}{\theta} - 1$, we obtain
    \begin{align*}
    W(e^{\frac{a_k - b_k}{\theta}+\frac{c_{\neg 1,k}-c_1}{\theta\beta} -1 }) = \frac{c_{\neg 1,k}-c_1}{\theta\beta} -\frac{y_k^\star}{\theta} - 1
    \\
    \Rightarrow y_k^\star = \frac{c_{\neg 1,k}-c_1}{\beta} - \theta - \theta \cdot W(e^{\frac{a_k - b_k}{\theta}+\frac{c_{\neg 1,k}-c_1}{\theta\beta} -1 })
    \end{align*}

    Summarizing, we obtain that, if~\eqref{eq:cond-y-pos} holds, then
    \begin{equation}
    y_k^\star =
    \frac{c_{\neg 1,k}-c_1}{\beta} - \theta \cdot\left(1+ W(e^{\frac{a_k - b_k}{\theta}+\frac{c_{\neg 1,k}-c_1}{\theta\beta} -1 }) \right)
    >0    
    \end{equation}

    If, instead, \eqref{eq:cond-y-pos} does not hold, then~$ y_k^\star=0$ (see~\eqref{eq:minimum-at-zero}). Therefore, 
    \begin{equation}
    y_k^\star =
    \max\left(
    \frac{c_{\neg 1,k}-c_1}{\beta} - \theta \cdot\left(1+ W(e^{\frac{a_k - b_k}{\theta}+\frac{c_{\neg 1,k}-c_1}{\theta\beta} -1 }) \right), 0
    \right)
    \end{equation}

\end{proof}

\begin{corollary*} \textbf{\ref{cor:incentive-comparison}.}
    Let $y^*_k$ be the optimal incentives as defined in Proposition \ref{prop:y_star}. Then $y^*_k$ is non-increasing in $a_k-b_k$.
\end{corollary*}
\begin{proof}
    By Proposition \ref{prop:y_star},
    \begin{equation}
        y_k^\star=\max\left(\frac{c_{\neg 1,k}-c_1}{\beta}-\theta \left(1+W(e^{\frac{a_k-b_k}{\theta}+\frac{c_{\neg 1,k}-c_1}{\theta \beta}-1}) \right), 0\right) \nonumber
    \end{equation}
    First, note that $e^{\frac{a_k-b_k}{\theta}+\frac{c_{\neg 1,k}-c_1}{\theta \beta}-1}$ is increasing with $a_k-b_k$, because of the term $\frac{a_k-b_k}{\theta}$ inside the exponential. 
    Following this, 
    $$ W(e^{\frac{a_k-b_k}{\theta}+\frac{c_{\neg 1,k}-c_1}{\theta \beta}-1}) $$ is increasing with $a_k-b_k$,
    since if the argument of the Lambert W function is positive and increasing, the Lambert W function is increasing as well.
    It follows that 
    $$-\theta \left(1+W(e^{\frac{a_k-b_k}{\theta}+\frac{c_{\neg 1,k}-c_1}{\theta \beta}-1}) \right)$$
    is decreasing with $a_k-b_k$ and so is the whole expression
    \begin{align*}
    \frac{c_{\neg 1,k}-c_1}{\beta}-\theta \left(1+W(e^{\frac{a_k-b_k}{\theta}+\frac{c_{\neg 1,k}-c_1}{\theta \beta}-1}) \right)
    \end{align*}
    Due to the fact that $y_k^\star = 0$ if $a_k-b_k$ is sufficiently large (see \eqref{eq:cond-y-pos}), it follows that $y_k^\star$ is non-increasing with $a_k-b_k$, which concludes the proof.
\end{proof}

\subsubsection*{Proof of Proposition \ref{prop:bound-y-star}}

\begin{prop*} \textbf{\ref{prop:bound-y-star}.}
    The optimal incentive is upper-bounded, i.e., $$0 \leq y_k^\star < \frac{c_{\neg 1,k}-c_1}{\beta}.$$
\end{prop*}
\begin{proof}
    By Proposition \ref{prop:y_star}, it suffices to show that 
    \begin{equation}
        \frac{c_{\neg 1,k}-c_1}{\beta}-\theta (1+W(e^{\frac{a_k-b_k}{\theta}+\frac{c_{\neg 1,k}-c_1}{\theta \beta}-1}) < \frac{c_{\neg 1,k}-c_1}{\beta}
    \end{equation}
    i.e.,
    \begin{equation}
        -\theta (1+W(e^{\frac{a_k-b_k}{\theta}+\frac{c_{\neg 1,k}-c_1}{\theta \beta}-1}) < 0
    \end{equation}
    \begin{equation}
        \iff 1+W(e^{\frac{a_k-b_k}{\theta}+\frac{c_{\neg 1,k}-c_1}{\theta \beta}-1}) > 0
    \end{equation}
    \begin{equation}
        \iff W(e^{\frac{a_k-b_k}{\theta}+\frac{c_{\neg 1,k}-c_1}{\theta \beta}-1}) > -1
    \end{equation}
    This is true, since if the argument of the Lambert W function is positive, the Lambert W function is positive as well. 
\end{proof}

\subsubsection*{Proof of Lemma~\ref{lem:InstPE}}

\begin{lemma*}
\begin{align}
\nonumber
\textit{GM}_r^{t,\textit{InstPE}} = 
\mathbb E_0[U_0(y_r^\star)] -
\mathbb E_0[U_0(y_0^\star)],
&& \textit{CM}_r^t=0,
&& \forall t
\end{align}
\end{lemma*}
\begin{proof}
The following equalities follow from a straightforward application of Obs.~\ref{obs:instPE}:
\begin{align}
\textit{GM}_r^{t,\textit{InstPE}}
& = \sum_{k\in\mathcal K} \mathbb P_r^\textit{InstPE}[Y_t=y_k^\star] 
\mathbb E_r[U_0(y_k^\star)]
\\ &
-
\sum_{k\in\mathcal K} \mathbb P_0^\textit{InstPE}[Y_t=y_k^\star] 
\mathbb E_r[U_0(y_k^\star)]
\\ &= 
\mathbb E_r[U_0(y_r^\star)] -
\mathbb E_r[U_0(y_0^\star)] 
\\ &= 
\mathbb E_0[U_0(y_r^\star)] -
\mathbb E_0[U_0(y_0^\star)]
\end{align}
and
\begin{align}
\nonumber
\textit{CM}_r^\textit{InstPE}(y)
=
\mathbb E_0[U_0(y)]-\mathbb E_0[U_0(y)]=0
\implies
\textit{CM}_r^t = 0.
\end{align}
\end{proof}


\section{Preference Elicitation via a Maximum-Likelihood Estimator}
\label{section:preferenc-elicitation}

We now specify the deterministic profile-inference rule~$\psi(\cdot)$
used in the numerical experiments.
At every timeslot~$s$, the agent chooses among all alternatives
$i^s\in\mathcal I$ according to the multinomial-logit model in
Eq.~\eqref{eq:logit}.
For online preference elicitation, however, the regulator uses only
whether the incentivized alternative~$1$ was selected.
Define the binary random variable and its realization as
\begin{align}
Z^s
&:=\mathbf{1}\{I^s=1\},
&
z^s
&:=\mathbf{1}\{i^s=1\}.
\label{eq:binary-observation}
\end{align}
Thus, $z^s=1$ when the agent chooses alternative~$1$, whereas
$z^s=0$ aggregates all choices $i^s>1$. This binary observation
model is plausible in applications where the regulator need not
collect the agent's complete choice. In particular, the data-minimization
principle of the General Data Protection Regulation (GDPR) supports
limiting the collected information to what is necessary for the
specified purpose. Accordingly, an implementation may reveal only
whether the incentivized alternative was selected---the information
needed to determine eligibility for the incentive---rather than which
non-incentivized alternative was chosen.
Let
\[
\mathbf h^{t-1}
=
(y^1,i^1,\ldots,y^{t-1},i^{t-1})
\]
be the realized history before timeslot~$t$.
Although the complete alternatives~$i^s$ appear in the history,
the inference rule depends on them only through the binary
observations~$z^s$.
Conditional on candidate profile~$k$ and incentive~$y^s$, the
binary observation follows
\begin{align}
Z^s
\mid
\mathbf H^{s-1}=\mathbf h^{s-1}
\sim
\operatorname{Ber}\!\left(p_k(y^s)\right),
\label{eq:binary-choice-distribution}
\end{align}
where $p_k(y^s)$ is the probability of choosing alternative~$1$
defined in Lemma~\ref{lem:u-profiles}.

The likelihood of the binary observations induced by
$\mathbf h^{t-1}$ under candidate profile~$k$ is therefore
\begin{align}
L_k^t(\mathbf h^{t-1})
&:=
\prod_{s=1}^{t-1}
\left[p_k(y^s)\right]^{z^s}
\left[1-p_k(y^s)\right]^{1-z^s}
\notag
\\
&=
L_k^{t-1}(\mathbf h^{t-2})
\left[p_k(y^{t-1})\right]^{z^{t-1}}
\left[1-p_k(y^{t-1})\right]^{1-z^{t-1}},
\label{eq:likelihood}
\end{align}
with $L_k^1=1$.
The proposed incentives do not introduce an additional factor in
the likelihood because each $y^s$ is a deterministic function of
the preceding observed history.

To avoid numerical underflow, the implementation maintains the
log-likelihood
\begin{align}
\ln L_k^t(\mathbf h^{t-1})
&=
\sum_{s=1}^{t-1}
\left[
z^s\ln p_k(y^s)
+
(1-z^s)\ln\!\left(1-p_k(y^s)\right)
\right]
\notag
\\
&=
\ln L_k^{t-1}(\mathbf h^{t-2})
+
z^{t-1}\ln p_k(y^{t-1})
\notag
\\
&\quad
+
(1-z^{t-1})
\ln\!\left(1-p_k(y^{t-1})\right),
\label{eq:loglikelihood}
\end{align}
with $\ln L_k^1=0$.

Let the set of maximum-likelihood profiles be
\begin{align}
\mathcal M^t(\mathbf h^{t-1})
:=
\arg\max_{k\in\mathcal K}
\ln L_k^t(\mathbf h^{t-1}).
\end{align}
The regulator selects, among the profiles in
$\mathcal M^t(\mathbf h^{t-1})$, the profile associated with the
smallest optimal incentive.
If several such profiles have the same optimal incentive, the
smallest profile index is selected.
Consequently, the deterministic inference rule used in the
numerical experiments is
\begin{align}
\hat k^t
=
\psi(\mathbf h^{t-1})
:=
\min
\left\{
k\in\mathcal M^t(\mathbf h^{t-1})
:
y_k^\star
=
\min_{j\in\mathcal M^t(\mathbf h^{t-1})}
y_j^\star
\right\}.
\label{eq:argmax-logL}
\end{align}

This online Bernoulli likelihood must not be confused with the
multinomial likelihood used for the offline calibration in
Sect.~\ref{sec:dataset-and-calibration}.
Accordingly, the Bernoulli distributions in
Eq.~\eqref{eq:bernoulli} concern the binary variable~$Z^s$, rather
than the complete categorical alternative~$I^s$.


\section{Appendix to the Case Study}

\providecommand{\fnref}[1]{\textsuperscript{\ref{#1}}}

\subsection{Dataset and Calibration Procedure}
\label{sec:dataset-and-calibration}

\subsubsection{Data}

The London Passenger Mode Choice (LPMC) dataset combines trip diaries from the
London Travel Demand Survey with the routes, travel times, public-transport
fares, driving fuel costs, and congestion charges associated with walking,
bike, public transport (PT), and car.\footnote{\label{fn:Hillel2018LPMC}Tim
Hillel, Mohammed Z. E. B. Elshafie, and Ying Jin. 2018. ``Recreating Passenger
Mode Choice-Sets for Transport Simulation: A Case Study of London, UK.''
\emph{Proceedings of the Institution of Civil Engineers--Smart Infrastructure
and Construction} 171(1):29--42. See Sec.~2.1, pp.~30--33.
\url{https://doi.org/10.1680/jsmic.17.00018}.} The public release contains
81,086 trips collected between April 2012 and March 2015.\footnote{\label{fn:BiogemeLPMC}Tim
Hillel. 2019. \emph{London Passenger Mode Choice}. Transport and Mobility
Laboratory, EPFL. See Table~1, p.~3.
\url{https://transp-or.epfl.ch/documents/technicalReports/CS_LPMC.pdf}.}

We retain travelers aged at least 12, since they are assumed to be responsive
to monetary incentives, and trips between home and work, between home and
educational activities, or between home and another activity. These are the
home-based work, home-based education, and home-based other purposes in the
LPMC documentation\fnref{fn:BiogemeLPMC}. Only trips made by bike,
PT, or car are retained. For each trip, the following information concerning
the trip or the traveler is reported: the travel time associated with each
mode; out-of-pocket costs for PT and car; the number of PT interchanges;
possession of a driving licence; car ownership, measured as the number of cars
in the household divided by the number of people in the household; trip
distance; and the bus and rail portions of PT in-vehicle time. Total PT time is
the sum of access, bus, rail, and interchange time. Trips missing a quantity
required by the utility specification are discarded.

Let $\mathcal D_{\mathrm{cal}}$ denote the calibration set and
$\mathcal D_{\mathrm{sim}}$ the simulation set. Among the retained bike trips,
70\% are assigned to $\mathcal D_{\mathrm{cal}}$ and 30\% to
$\mathcal D_{\mathrm{sim}}$; the same split is performed independently among
PT trips and among car trips. Hence, the proportions of bike, PT, and car trips
are preserved in both sets. The resulting set $\mathcal D_{\mathrm{cal}}$
contains 34,198 trips made by 20,577 travelers, whereas
$\mathcal D_{\mathrm{sim}}$ contains 14,656 trips made by 11,709 travelers.

For a trip set $S\in\{\mathcal D_{\mathrm{cal}},\mathcal D_{\mathrm{sim}}\}$,
let $\mathcal N_S$ denote its travelers, $\bar d_S$ its mean trip distance,
and $\omega_{\mathrm{bus},S}$ the share of bus time in total bus-plus-rail
in-vehicle time. Table~\ref{tab:appnum-direct-data} separates quantities
computed directly from the trip records from quantities estimated by the
calibration procedure.

\begin{table}[t]
\centering
\footnotesize
\setlength{\tabcolsep}{3.1pt}
\caption{Quantities computed directly from the calibration and simulation sets. A dash indicates that the quantity is not used for that set: $\bar d_S$ and $\omega_{\mathrm{bus},S}$ define the representative simulation context and are therefore computed from $\mathcal D_{\mathrm{sim}}$.}
\label{tab:appnum-direct-data}
\begin{tabular}{@{}lccc@{}}
\toprule
Quantity & Symbol & $\mathcal D_{\mathrm{cal}}$ & $\mathcal D_{\mathrm{sim}}$\\
\midrule
Trips & $|S|$ & 34,198 & 14,656\\
Travelers & $|\mathcal N_S|$ & 20,577 & 11,709\\
Mean trip distance (km/trip) & $\bar d_S$ & -- & 5.41\\
Bus-time share & $\omega_{\mathrm{bus},S}$ & -- & 0.65\\
\bottomrule
\end{tabular}
\end{table}

\subsubsection{Profile-based choice model}

Let $\mathcal K=\{0,1,2\}$ be the set of possible true profiles. After
calibration, the profiles are labeled so that $k=0$ is bike-oriented, $k=1$ is
PT-oriented, and $k=2$ is car-oriented. For every $k\in\mathcal K$, the
systematic utility $u_k(i,y)$ introduced in the main paper is instantiated for
trip $\tau$ as
\begin{equation}
\begin{aligned}
u_{k,\tau}(1,y)
  &= \xi_k^{\mathrm{bike}}-\nu t_{\tau}^{\mathrm{bike}}+y,\\
u_{k,\tau}(2,y)
  &= \xi_k^{\mathrm{PT}}-\nu t_{\tau}^{\mathrm{PT}}
     -C_{\tau}^{\mathrm{PT}}-\iota_k m_{\tau},\\
u_{k,\tau}(3,y)
  &= -\nu t_{\tau}^{\mathrm{car}}-C_{\tau}^{\mathrm{car}}
     +\zeta_k^{\mathrm{lic}}\chi_{\tau}^{\mathrm{lic}}
     +\zeta_k^{\mathrm{own}}\chi_{\tau}^{\mathrm{own}}.
\end{aligned}
\label{eq:appnum-systematic-utility}
\end{equation}

\begin{table}[t]
    \centering
    \small
    \setlength{\tabcolsep}{3pt}
    \begin{tabular}{@{}p{0.27\columnwidth}
        p{\dimexpr 0.73\columnwidth-2\tabcolsep\relax}@{}}
        \toprule
        Symbol & Definition \\
        \midrule

        $\tau$
        & Index of a recorded trip. \\

        $i\in\{1,2,3\}$
        & Travel alternative: $i=1$ denotes bike, $i=2$ public transport
        (PT), and $i=3$ car. \\

        $t_\tau^i$
        & Travel time of alternative $i$ for trip $\tau$, in hours. \\

        $C_\tau^{\mathrm{PT}}$
        & Out-of-pocket PT cost for trip $\tau$. \\

        $C_\tau^{\mathrm{car}}$
        & Out-of-pocket car cost for trip $\tau$. \\

        $m_\tau$
        & Number of PT interchanges for trip $\tau$. \\

        $\chi_\tau^{\mathrm{lic}}$
        & Indicator equal to $1$ if the traveler making trip $\tau$ holds a
        driving licence, and $0$ otherwise. \\

        $\chi_\tau^{\mathrm{own}}$
        & Car ownership of the traveler making trip $\tau$, measured as the
        number of cars in the household divided by the number of people in
        the household. \\

        $\nu$
        & Common value of travel time, in pounds per hour. \\

        $\iota_k$
        & Profile-$k$ utility penalty for each additional PT interchange. \\

        $\zeta_k^{\mathrm{lic}}$
        & Profile-$k$ increment in car utility associated with holding a
        driving licence. \\

        $\zeta_k^{\mathrm{own}}$
        & Profile-$k$ increment in car utility per unit of
        $\chi_\tau^{\mathrm{own}}$. \\

        $\xi_k^{\mathrm{bike}}$
        & Intrinsic profile-$k$ preference for bike not represented by the
        other trip or traveler characteristics. \\

        $\xi_k^{\mathrm{PT}}$
        & Intrinsic profile-$k$ preference for PT not represented by the
        other trip or traveler characteristics. \\

        \bottomrule
    \end{tabular}
    \caption{Notation used in the trip-specific systematic utilities.}
    \label{tab:trip-utility-notation}
\end{table}

The notation used in the trip-specific utility functions is summarized in
Table~\ref{tab:trip-utility-notation}.

For any alternative $i$, trip $\tau$, profile $k$, and incentive $y$, let
$\mathbb P_i^{u_{k,\tau}}(y)$ denote the probability in
Eq.~\eqref{eq:logit} evaluated with the trip-specific systematic utility
function $u_{k,\tau}$ in Eq.~\eqref{eq:appnum-systematic-utility}. In
particular, $\mathbb P_{i_\tau}^{u_{k,\tau}}(0)$ is the probability that
profile $k$ chooses the mode $i_\tau$ recorded for trip $\tau$ when no bike
incentive is offered. The argument $0$ is used during calibration because the
trips in the dataset were recorded without the simulated incentive policy.

\subsubsection{Calibration procedure}

The offline calibration described in this subsection is distinct from the
regulator's online learning process in
Appendix~\ref{section:preferenc-elicitation}. The offline calibration is used
only to obtain behavioral parameters for the numerical experiment. Its
log-likelihood $\ell_{nk}^{\mathrm{cal}}$, defined below from the LPMC trips,
must therefore not be confused with the regulator's likelihood
$L_k^t(\mathbf h^{t-1})$ in
Eqs.~\eqref{eq:likelihood}--\eqref{eq:loglikelihood}.

Let $\mathcal N_{\mathrm{cal}}$ be the travelers represented in
$\mathcal D_{\mathrm{cal}}$, let $\mathcal T_n^{\mathrm{cal}}$ be the set of
calibration trips made by traveler $n\in\mathcal N_{\mathrm{cal}}$, and let
$i_\tau$ be the mode chosen on trip $\tau$. For parameter vector $\Theta$, the
offline calibration log-likelihood of traveler $n$ under candidate true
profile $k$ is
\begin{equation}
\ell_{nk}^{\mathrm{cal}}(\Theta)
=\sum_{\tau\in\mathcal T_n^{\mathrm{cal}}}
\log \mathbb P_{i_\tau}^{u_{k,\tau}}(0).
\label{eq:appnum-calibration-loglikelihood}
\end{equation}
The trip and traveler quantities in
Eq.~\eqref{eq:appnum-systematic-utility} are fixed inputs. The parameters to be
calibrated are
\begin{equation}
\Theta=
\left(
\nu,\theta,
\left\{
\xi_k^{\mathrm{bike}},\xi_k^{\mathrm{PT}},\iota_k,
\zeta_k^{\mathrm{lic}},\zeta_k^{\mathrm{own}}
\right\}_{k\in\mathcal K}
\right),
\label{eq:appnum-parameter-vector}
\end{equation}
where $\theta$ is the Gumbel scale in Eq.~\eqref{eq:logit}.

The initial profile-membership probabilities are based on each traveler's
mode shares. Let $\mu_{nk}$ be the fraction of trips in
$\mathcal T_n^{\mathrm{cal}}$ made using the mode associated with profile $k$:
bike for $k=0$, PT for $k=1$, and car for $k=2$. Before enforcing equal total
membership across profiles, the initial profile-membership probabilities are
\begin{equation}
\widetilde q_{nk}^{(0)}=
\frac{\exp(\eta_{\mathrm{init}}\mu_{nk})}
{\sum_{j\in\mathcal K}\exp(\eta_{\mathrm{init}}\mu_{nj})},
\qquad \eta_{\mathrm{init}}=2.
\label{eq:appnum-initial-membership}
\end{equation}
The mode-share scaling parameter $\eta_{\mathrm{init}}$ Controls how strongly observed mode shares influence the initial profile probabilities. A larger value gives more probability to profiles associated with frequently used modes; 0 would give equal probability to every profile.. We heuristically set~$\eta_{\mathrm{init}}=2$.

At calibration step $d\geq1$, the parameters $\Theta^{(d)}$ determine the
pre-adjustment profile-membership probabilities
\begin{equation}
\widetilde q_{nk}^{(d)}=
\frac{\exp\!\left(\ell_{nk}^{\mathrm{cal}}(\Theta^{(d)})\right)}
{\sum_{j\in\mathcal K}
 \exp\!\left(\ell_{nj}^{\mathrm{cal}}(\Theta^{(d)})\right)}.
\label{eq:appnum-pre-adjustment-membership}
\end{equation}
The superscript $(0)$ in
Eq.~\eqref{eq:appnum-initial-membership} denotes initialization from empirical
mode shares, whereas the superscript $(d)$ in
Eq.~\eqref{eq:appnum-pre-adjustment-membership} denotes calibration step $d$.
The tilde identifies probabilities before imposing the profile-size condition
in Eq.~\eqref{eq:appnum-equal-profile-membership}.

For every $d\in\{0,\ldots,N_{\mathrm{step}}\}$, adjustment parameters
$\lambda_k^{(d)}\in\mathbb R$ are selected so that
\begin{align}
q_{nk}^{(d)}
&=
\frac{\widetilde q_{nk}^{(d)}\exp(\lambda_k^{(d)})}
{\sum_{j\in\mathcal K}
 \widetilde q_{nj}^{(d)}\exp(\lambda_j^{(d)})},
\label{eq:appnum-adjusted-membership}\\
\sum_{n\in\mathcal N_{\mathrm{cal}}}q_{nk}^{(d)}
&=\frac{|\mathcal N_{\mathrm{cal}}|}{|\mathcal K|},
\qquad k\in\mathcal K.
\label{eq:appnum-equal-profile-membership}
\end{align}
The quantity $q_{nk}^{(d)}$ is the membership probability assigned to traveler
$n$ and profile $k$ at calibration step $d$. Therefore,
$\sum_{n\in\mathcal N_{\mathrm{cal}}}q_{nk}^{(d)}$ is the fraction of agents
associated to profile $k$ within the calibration set, and
Eq.~\eqref{eq:appnum-equal-profile-membership} sets this quantity to one third of $|\mathcal N_{\mathrm{cal}}|$ for every $k\in\mathcal K$. We do so to prevent calibration to favor one profile over the others.\footnote{
We also run calibration imposing that the fraction of trips (and not of agents) be 1/3 per each profile. Final result did not change significantly and are thus omitted.
}
Observe that the adjustment parameters $\lambda_k^{(d)}$ shift the log-membership scores in Eq.~\eqref{eq:appnum-adjusted-membership}. Adding the same constant to every $\lambda_k^{(d)}$ leaves the probabilities unchanged; hence, we set
$\lambda_0^{(d)}=0$ and solve the two equations in
Eq.~\eqref{eq:appnum-equal-profile-membership} corresponding to $k=1$ and
$k=2$ for $\lambda_1^{(d)}$ and $\lambda_2^{(d)}$.

For fixed profile-membership probabilities $q_{nk}^{(d-1)}$, we find the
parameters $\Theta$ that minimize
\begin{align}
\mathcal J(\Theta;q^{(d-1)})
&=\mathcal J_{\mathrm{NLL}}(\Theta;q^{(d-1)})
  +\rho\,\mathcal R(\Theta),
\label{eq:appnum-calibration-loss}\\
\mathcal J_{\mathrm{NLL}}(\Theta;q^{(d-1)})
&=-\frac{1}{|\mathcal D_{\mathrm{cal}}|}
\sum_{n\in\mathcal N_{\mathrm{cal}}}
\sum_{\tau\in\mathcal T_n^{\mathrm{cal}}}
\nonumber\\[-1mm]
&\hspace{8mm}
\sum_{k\in\mathcal K}q_{nk}^{(d-1)}
\log \mathbb P_{i_\tau}^{u_{k,\tau}}(0).
\label{eq:appnum-weighted-nll}
\end{align}
Here, $|\mathcal D_{\mathrm{cal}}|$ is the number of calibration trips and
$\mathcal J_{\mathrm{NLL}}$ in
Eq.~\eqref{eq:appnum-weighted-nll} is the average negative log-likelihood of
the modes chosen in each recorded trip, weighted by the profile-membership
probabilities $q_{nk}^{(d-1)}$. The second term in
Eq.~\eqref{eq:appnum-calibration-loss} is the regularization term, with
\begin{align}
\mathcal R(\Theta)=\sum_{k\in\mathcal K}\Big[&
(\xi_k^{\mathrm{bike}}-\bar\xi^{\mathrm{bike}})^2+
(\xi_k^{\mathrm{PT}}-\bar\xi^{\mathrm{PT}})^2
\nonumber\\[-1mm]
&+(\iota_k-\bar\iota)^2+
(\zeta_k^{\mathrm{lic}}-\bar\zeta^{\mathrm{lic}})^2
\nonumber\\[-1mm]
&+(\zeta_k^{\mathrm{own}}-\bar\zeta^{\mathrm{own}})^2\Big],
\label{eq:appnum-regularizer}
\end{align}
where each barred quantity is the mean of the corresponding coefficient over
$k\in\mathcal K$. Because both the traveler memberships $q_{nk}$ and the
profile-specific coefficients in Eq.~\eqref{eq:appnum-parameter-vector} are
estimated, markedly different coefficient vectors can produce similar values
of the weighted likelihood in Eq.~\eqref{eq:appnum-weighted-nll}. The
regularization term in Eq.~\eqref{eq:appnum-regularizer} favors the less
extreme of such parameterizations by penalizing cross-profile dispersion. The
coefficient $\rho=0.01$ gives this term a small weight relative to the average
negative log-likelihood. The common value of time $\nu$ and common Gumbel scale
$\theta$ are calibrated but not regularized because each has a single value
shared by all profiles; neither contributes to cross-profile dispersion.

The calibration has $N_{\mathrm{step}}=4$ steps. At calibration step $d$, we
first initialize $\Theta$ at $\Theta^{(d-1)}$; for $d=1$, the initial values
$\Theta^{(0)}$ in Table~\ref{tab:appnum-initial-bounds} are used. We then find
$\Theta^{(d)}$ that minimizes Eq.~\eqref{eq:appnum-calibration-loss}, holding
$q_{nk}^{(d-1)}$ fixed and enforcing the bounds in
Table~\ref{tab:appnum-initial-bounds}. To do so, we use the analytic gradient
of loss~\eqref{eq:appnum-calibration-loss} with respect to parameters $\Theta$
and the limited-memory Broyden--Fletcher--Goldfarb--Shanno algorithm for
bound-constrained optimization (L-BFGS-B).\footnote{\label{fn:Byrd1995LBFGSB}Richard H.
Byrd, Peihuang Lu, Jorge Nocedal, and Ciyou Zhu. 1995. ``A Limited Memory
Algorithm for Bound Constrained Optimization.'' \emph{SIAM Journal on
Scientific Computing} 16(5):1190--1208. See Sec.~2, pp.~1193--1198.
\url{https://doi.org/10.1137/0916069}.} The resulting $\Theta^{(d)}$ is substituted in
Eqs.~\eqref{eq:appnum-calibration-loglikelihood} and
\eqref{eq:appnum-pre-adjustment-membership}; the adjustment in
Eqs.~\eqref{eq:appnum-adjusted-membership}--\eqref{eq:appnum-equal-profile-membership}
then gives $q_{nk}^{(d)}$.

Profile 0 is required to satisfy
\begin{equation}
a_0-b_0\geq\pounds0.15.
\label{eq:appnum-profile0-margin}
\end{equation}
This condition represents a profile for which the zero-incentive bike utility
is slightly larger than the inclusive utility of the non-bike alternatives in
the representative travel context. The value \pounds0.15 is a
small separation margin rather than an empirical estimate. At the calibrated
scale $\widehat\theta=\pounds3.68$, the boundary in
Eq.~\eqref{eq:appnum-profile0-margin} implies, by Eq.~\eqref{eq:pk},
\begin{equation}
p_0(0)\geq
\sigma\!\left(\frac{0.15}{3.68}\right)=51.02\%.
\label{eq:appnum-profile0-probability-margin}
\end{equation}
Thus, the condition requires only a 1.02-percentage-point preference above
indifference between bike and the aggregate non-bike option. It identifies a
strictly bike-oriented profile without imposing a large behavioral
separation. The calibrated value is
$\widehat a_0-\widehat b_0=\pounds0.36$, which is above the boundary in
Eq.~\eqref{eq:appnum-profile0-margin}. The quantities $a_0$ and $b_0$ are
defined as in Eq.~\eqref{eq:pk}, using the profile-weighted trip attributes at
the current calibration step. After estimating $\Theta^{(d)}$ and again after
obtaining $q_{nk}^{(d)}$, the bike constant is updated according to
\begin{equation}
\xi_0^{\mathrm{bike}}\leftarrow
\xi_0^{\mathrm{bike}}+
\max\{0,\pounds0.15-(a_0-b_0)\},
\label{eq:appnum-profile0-projection}
\end{equation}
which restores Eq.~\eqref{eq:appnum-profile0-margin} whenever that condition is
violated.

\begin{table}[t]
\centering
\small
\setlength{\tabcolsep}{3.2pt}
\caption{Settings of the offline calibration procedure.}
\label{tab:appnum-hyperparameters}
\begin{tabular}{@{}lll@{}}
\toprule
Quantity & Symbol & Value\\
\midrule
Number of profiles & $|\mathcal K|$ & 3\\
Mode-share scaling in Eq.~\eqref{eq:appnum-initial-membership}
 & $\eta_{\mathrm{init}}$ & 2\\
Calibration steps & $N_{\mathrm{step}}$ & 4\\
Maximum L-BFGS-B iterations per step & $N_{\mathrm{opt}}$ & 80\\
Function tolerance & $\varepsilon_{\mathrm{fun}}$ & $10^{-10}$\\
Gradient tolerance & $\varepsilon_{\mathrm{grad}}$ & $10^{-5}$\\
Maximum line-search steps & $N_{\mathrm{line}}$ & 30\\
Regularization coefficient & $\rho$ & 0.01\\
Profile-0 margin in Eq.~\eqref{eq:appnum-profile0-margin}
 & $a_0-b_0$ & at least \pounds0.15\\
\bottomrule
\end{tabular}
\end{table}

\begin{table}[t]
\centering
\scriptsize
\setlength{\tabcolsep}{2.1pt}
\caption{Initialization, admissible ranges, and calibrated values.
Profile-indexed values are ordered as $k=0,1,2$. The calibrated
profile-specific coefficients enter the simulations through the constructed
utilities $a_k$ and $b_k$ reported in
Table~\ref{tab:appnum-calibrated}.}
\label{tab:appnum-initial-bounds}
\begin{tabular}{@{}p{3.25cm}ccc@{}}
\toprule
Parameter & $\Theta^{(0)}$ & Range & Calibrated\\
\midrule

$\xi_k^{\mathrm{bike}}$: intrinsic bike preference (\pounds)
 & $(5,1,-1)$ & $[-15,15]$ & through $a_k$\\

$\xi_k^{\mathrm{PT}}$: intrinsic PT preference (\pounds)
 & $(0,2,-1)$ & $[-15,15]$ & through $b_k$\\

$\nu$: value of time (\pounds/h)
 & 10 & $[3,30]$ & 10.48\\

$\iota_k$: PT-interchange penalty (\pounds/interchange)
 & $(1,1,1)$ & $[0,8]$ & through $b_k$\\

$\zeta_k^{\mathrm{lic}}$: driving-licence utility increment (\pounds)
 & $(2,2,4)$ & $[0,10]$ & through $b_k$\\

$\zeta_k^{\mathrm{own}}$: car-ownership utility increment (\pounds/unit)
 & $(2,2,4)$ & $[0,10]$ & through $b_k$\\

$\theta$: Gumbel scale (\pounds)
 & 3.50 & $[1,12]$ & 3.68\\

\midrule

$a_k$: constructed bike utility (\pounds)
 & -- & -- & $(6.12,-6.78,-6.59)$\\

$b_k$: constructed non-bike utility (\pounds)
 & -- & -- & $(5.76,4.94,5.35)$\\

$a_k-b_k$: constructed utility difference (\pounds)
 & -- & Eq.~\eqref{eq:appnum-profile0-margin}
 & $(0.36,-11.73,-11.93)$\\

\bottomrule
\end{tabular}
\end{table}

The mixture negative log-likelihood used to evaluate descriptive fit is
\begin{equation}
\mathcal J_{\mathrm{mix}}(\Theta)=
-\frac{1}{|\mathcal D_{\mathrm{cal}}|}
\sum_{n\in\mathcal N_{\mathrm{cal}}}
\log\!\left[
\sum_{k\in\mathcal K}\frac{1}{|\mathcal K|}
\exp\!\left(\ell_{nk}^{\mathrm{cal}}(\Theta)\right)
\right].
\label{eq:appnum-mixture-nll}
\end{equation}
A nat is the unit of log loss obtained when the natural logarithm is used. The
fitted value of $\mathcal J_{\mathrm{mix}}$ in
Eq.~\eqref{eq:appnum-mixture-nll} is 0.67 nat/trip. Uniform choice gives 1.10
nat/trip, and assigning each mode its aggregate frequency in
$\mathcal D_{\mathrm{cal}}$ gives 0.83 nat/trip. For a model $M$ with average
negative log-likelihood $\mathcal J_M$, define the geometric mean probability
assigned to the modes chosen in the recorded trips as
\begin{equation}
\overline p_{\mathrm{geom}}^{M}=\exp(-\mathcal J_M).
\label{eq:appnum-geometric-probability}
\end{equation}
Using the unrounded losses gives
\begin{equation}
\left(
\overline p_{\mathrm{geom}}^{\mathrm{fit}},
\overline p_{\mathrm{geom}}^{\mathrm{shares}},
\overline p_{\mathrm{geom}}^{\mathrm{uniform}}
\right)
=(51.12\%,43.74\%,33.33\%).
\label{eq:appnum-geometric-probability-values}
\end{equation}
The fitted profile model therefore assigns 7.38 percentage points more
geometric-mean probability to the realized choices than the aggregate-share
baseline and 17.79 percentage points more than uniform choice. Hence, the
profile-specific utilities capture systematic variation in mode choice beyond
the aggregate frequencies of the three modes.

\subsubsection{Calibrated parameter values, profile construction and context parameters}

Let $\mathcal N_{\mathrm{sim}}$ be the travelers represented in
$\mathcal D_{\mathrm{sim}}$, and let $\mathcal T_n^{\mathrm{sim}}$ be the trips
made by traveler $n\in\mathcal N_{\mathrm{sim}}$ in that set. The calibrated
parameters $\widehat\Theta$ give
\begin{equation}
\ell_{nk}^{\mathrm{sim}}(\widehat\Theta)
=\sum_{\tau\in\mathcal T_n^{\mathrm{sim}}}
\log\mathbb P_{i_\tau}^{u_{k,\tau}}(0).
\label{eq:appnum-simulation-loglikelihood}
\end{equation}
The simulation-set membership probabilities $q_{nk}^{\mathrm{sim}}$ are
obtained from Eq.~\eqref{eq:appnum-pre-adjustment-membership} by replacing
$\ell_{nk}^{\mathrm{cal}}(\Theta^{(d)})$ with
$\ell_{nk}^{\mathrm{sim}}(\widehat\Theta)$. The profile-size adjustment in
Eqs.~\eqref{eq:appnum-adjusted-membership}--\eqref{eq:appnum-equal-profile-membership}
is not applied to $\mathcal D_{\mathrm{sim}}$. The share of profile $k$ in the
simulation set $\mathcal D_{\mathrm{sim}}$ is
\begin{equation}
\pi_k=\frac{1}{|\mathcal N_{\mathrm{sim}}|}
\sum_{n\in\mathcal N_{\mathrm{sim}}}q_{nk}^{\mathrm{sim}}.
\label{eq:appnum-profile-share}
\end{equation}

For each trip $\tau\in\mathcal D_{\mathrm{sim}}$ and alternative $i$, define
\begin{equation}
\begin{aligned}
\mathbf x_{\tau i}=\big(&
\mathbb I_{\{i=1\}},\mathbb I_{\{i=2\}},t_\tau^i,C_\tau^i,
 m_\tau\mathbb I_{\{i=2\}},\\[-1mm]
&\chi_\tau^{\mathrm{lic}}\mathbb I_{\{i=3\}},
 \chi_\tau^{\mathrm{own}}\mathbb I_{\{i=3\}}
\big)^{\!\top},
\end{aligned}
\label{eq:appnum-trip-attribute-vector}
\end{equation}
where $C_\tau^1=0$, $C_\tau^2=C_\tau^{\mathrm{PT}}$, and
$C_\tau^3=C_\tau^{\mathrm{car}}$. Thus, $\mathbf x_{\tau i}$ contains the
alternative indicators, alternative-specific travel time, out-of-pocket cost,
number of PT interchanges, driving-licence indicator, and household car
ownership entering Eq.~\eqref{eq:appnum-systematic-utility}. Let $n(\tau)$
denote the traveler who made trip $\tau$. For every profile--alternative pair
$(k,i)$, the representative attribute vector is
\begin{equation}
\mathbf x_{i,k}=
\frac{\sum_{\tau\in\mathcal D_{\mathrm{sim}}}
 q_{n(\tau)k}^{\mathrm{sim}}\mathbf x_{\tau i}}
{\sum_{\tau\in\mathcal D_{\mathrm{sim}}}
 q_{n(\tau)k}^{\mathrm{sim}}}.
\label{eq:appnum-representative-attributes}
\end{equation}
The weight $q_{n(\tau)k}^{\mathrm{sim}}$ gives greater influence to trips made
by travelers whose profile-membership probability $q_{n(\tau)k}^{\mathrm{sim}}$
is higher. This ensures that the representative attribute vector
$\mathbf x_{i,k}$ is more strongly affected by profile memberships inferred
with higher confidence. Substituting $\mathbf x_{i,k}$ and
$\widehat\Theta$ into Eq.~\eqref{eq:appnum-systematic-utility} determines the
profile-level systematic utilities $u_k(i,0)$. Profiles are relabeled so that
$a_0-b_0\geq a_1-b_1\geq a_2-b_2$, consistently with the ordering in the main
paper.

Table~\ref{tab:appnum-calibrated} reports the common calibrated parameters
$\widehat\nu$ and $\widehat\theta$, the profile shares $\pi_k$, and the
profile-level quantities obtained from
Eqs.~\eqref{eq:appnum-profile-share}--\eqref{eq:appnum-representative-attributes}.
The zero-incentive mode probabilities $\mathbb P_i^{u_k}(0)$,
$i\in\{1,2,3\}$, are evaluated using Eq.~\eqref{eq:logit} with the utilities
$u_k(i,0)$.

\begin{table}[t]
\centering
\footnotesize
\setlength{\tabcolsep}{3.1pt}
\caption{Calibrated parameter values and profile-level context quantities.}
\label{tab:appnum-calibrated}
\begin{tabular}{@{}lrl@{}}
\toprule
Common quantity & Value & Unit\\
\midrule
Value of time $\widehat\nu$ & 10.48 & \pounds/h\\
Gumbel scale $\widehat\theta$ & 3.68 & \pounds\\
\bottomrule
\end{tabular}

\medskip
\begin{tabular}{@{}crrrrrr@{}}
\toprule
$k$ & $\pi_k$ & $\mathbb P_1^{u_k}(0)$ & $\mathbb P_2^{u_k}(0)$
& $\mathbb P_3^{u_k}(0)$ & $a_k$ & $b_k$\\
 & (\%) & (\%) & (\%) & (\%) & (\pounds) & (\pounds)\\
\midrule
0 & 21.26 & 52.44 & 19.12 & 28.44 & 6.12 & 5.76\\
1 & 39.36 & 3.96 & 52.08 & 43.96 & -6.78 & 4.94\\
2 & 39.38 & 3.75 & 43.10 & 53.15 & -6.59 & 5.35\\
\bottomrule
\end{tabular}
\end{table}

\begin{figure}[t]
\centering
\includegraphics[width=\columnwidth]{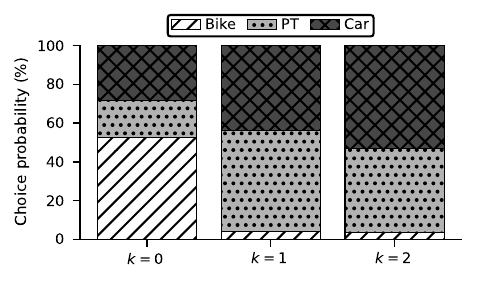}
\caption{Zero-incentive mode probabilities $\mathbb P_i^{u_k}(0)$ for the
three calibrated profiles $k\in\mathcal K$. Alternatives $i=1,2,3$ denote
bike, PT, and car, respectively.}
\label{fig:appnum-profile-probabilities}
\end{figure}

For clarity, before rounding, the representative profile-level
quantities used throughout the case study are
\[
\{a_k\}_{k=0}^{2}
=
(6.1178,-6.7834,-6.5886),
\]
\[
\{b_k\}_{k=0}^{2}
=
(5.7590,4.9421,5.3451),
\]
and therefore
\[
\{a_k-b_k\}_{k=0}^{2}
=
(0.3588,-11.7255,-11.9337).
\]

\subsubsection{Externalities and calibration of the policy weight $\beta$}

Table~\ref{tab:appnum-literature-inputs} reports the externally specified
quantities. Greenhouse-gas (GHG) factors are expressed as carbon-dioxide
equivalent (CO$_2$e), because the official conversion factors aggregate the
climate effects of the covered greenhouse gases into an equivalent mass of
CO$_2$.\footnote{\label{fn:DESNZ2026GHG}UK Department for Energy Security and Net
Zero. 2026. \emph{Greenhouse Gas Reporting: Conversion Factors 2026}.
See methodology Secs.~1 and~5 and the flat-file worksheet ``Factors by
Category.'' \url{https://www.gov.uk/government/publications/greenhouse-gas-reporting-conversion-factors-2026}.}
Bike is assigned zero direct operational emissions. The monetary carbon value
$V_{\rm CO_2e}$ follows UK government appraisal guidance.\footnote{\label{fn:DESNZ2021CarbonValues}UK
Department for Energy Security and Net Zero and UK Department for Business,
Energy and Industrial Strategy. 2021. \emph{Valuation of Greenhouse Gas
Emissions: For Policy Appraisal and Evaluation}. See Annex~1.
\url{https://www.gov.uk/government/publications/valuing-greenhouse-gas-emissions-in-policy-appraisal/valuation-of-greenhouse-gas-emissions-for-policy-appraisal-and-evaluation}.}
The road externalities $C_{\rm cong}$ and $C_{\rm other}$ are taken from the
Institute for Fiscal Studies assessment of motoring external costs.\footnote{\label{fn:AdamStroud2019}Stuart
Adam and Rebekah Stroud. 2019. ``A Road Map for Motoring Taxation.'' In
\emph{The IFS Green Budget: October 2019}. Institute for Fiscal Studies,
London. See Sec.~9.2 and Tables~9.1--9.2.
\url{https://ifs.org.uk/books/road-map-motoring-taxation}.}
The cycling payment rate $R_{\rm Cic}$ is taken from a field experiment.\footnote{\label{fn:Ciccone2021Cycling}Alice
Ciccone, Aslak Fyhri, and Hanne B. Sundf{\o}r. 2021. ``Using Behavioral
Insights to Incentivize Cycling: Results from a Field Experiment.''
\emph{Journal of Economic Behavior \& Organization} 188:1035--1058. See
Sec.~3.1, pp.~1038--1039.
\url{https://doi.org/10.1016/j.jebo.2021.06.011}.}
The exchange rate $X_{\rm NOK/GBP}$ is the 2021 annual average reported by DNB
Markets.\footnote{\label{fn:DNB2021Exchange}DNB Markets. 2021. \emph{Monthly and
Yearly Averages for Major Currencies 2021}. Annual-average GBP/NOK mid-rate;
the table states that the reported values average daily mid-rates.
\url{https://www.dnb.no/en/markets/fx-and-ir/fx-rates/monthly-yearly-averages/average-rates-2021?la=EN}.}

\begin{table*}[t]
\centering
\footnotesize
\setlength{\tabcolsep}{4pt}
\caption{Externally specified inputs. The reference column identifies the precise source location.}
\label{tab:appnum-literature-inputs}
\begin{tabular}{@{}p{3.25cm}cp{3.15cm}p{8.05cm}@{}}
\toprule
Input & Symbol & Value & Reference\\
\midrule
London local-bus emissions
& $E_{\rm bus}$ & 63.60 gCO$_2$e/passenger-km
& Category ``Local bus, London'' and methodology in [Footnote~\ref{fn:DESNZ2026GHG} - Sec.~5, pp.~74--76.]\\
National-rail emissions
& $E_{\rm rail}$ & 30.92 gCO$_2$e/passenger-km
& Category ``National rail'' and methodology in [Footnote~\ref{fn:DESNZ2026GHG} - Sec.~5, p.~78.]\\
Average-car emissions
& $E_{\rm car}$ & 165.91 gCO$_2$e/vehicle-km
& Category ``Average car'' and methodology in [Footnote~\ref{fn:DESNZ2026GHG} - Sec.~5, pp.~46--69].\\
Central carbon value
& $V_{\rm CO_2e}$ & \pounds264/tCO$_2$e
& Central 2026 value in [Footnote~\ref{fn:DESNZ2021CarbonValues} - Annex~1]\\
Congestion externality
& $C_{\rm cong}$ & \pounds0.81/vehicle-km
& Marginal congestion cost of one additional vehicle-kilometre on principal roads in London, estimated in [Footnote~\ref{fn:AdamStroud2019} Sec.~9.2, Table~9.2].\\
Other car externalities
& $C_{\rm other}$ & \pounds0.03/vehicle-km
& Sum of accident, local-air-pollution, noise, and infrastructure costs, estimated in [Footnote~\ref{fn:AdamStroud2019} Sec.~9.2, Table~9.2].\\
Reference biking incentive
& $R_{\rm Cic}$ & 2.00 NOK/cycled-km
& Pay-per-kilometre treatment, tested in a real world experiment - see [Footnote~\ref{fn:Ciccone2021Cycling} Sec.~3.1, pp.~1038--1039].
We use this value only to set the value of~$\beta$, as discussed in Eqs.~\eqref{eq:appnum-ciccone-nok}-\eqref{eq:appnum-beta}. We do not use it directly as the incentives proposed, which we instead calculate via Prop.~\ref{prop:y_star}.
\\
NOK/GBP exchange rate
& $X_{\rm NOK/GBP}$ & 11.83 NOK/\pounds
& 2021 annual-average GBP rate [Footnote~\ref{fn:DNB2021Exchange}].\\
\bottomrule
\end{tabular}
\end{table*}

Using the direct simulation-set quantities $\bar d_{\mathcal D_{\rm sim}}$
and $\omega_{\mathrm{bus},\mathcal D_{\rm sim}}$ in
Table~\ref{tab:appnum-direct-data}, the operational emissions of a
representative PT trip and a representative car trip are
\begin{align}
E_{\rm PT}^{\rm trip}
&=\bar d_{\mathcal D_{\rm sim}}
\left[\omega_{\mathrm{bus},\mathcal D_{\rm sim}}E_{\rm bus}
 +(1-\omega_{\mathrm{bus},\mathcal D_{\rm sim}})E_{\rm rail}\right],
\label{eq:appnum-pt-emissions}\\
E_{\rm car}^{\rm trip}
&=\bar d_{\mathcal D_{\rm sim}} E_{\rm car}.
\label{eq:appnum-car-emissions}
\end{align}
Thus, emissions from car trips $E_{\rm car}^{\rm trip}$ are obtained by
multiplying the official average-car vehicle-kilometre factor $E_{\rm car}$ by
the average kilometres traveled $\bar d_{\mathcal D_{\rm sim}}$ computed from
the simulation set $\mathcal D_{\rm sim}$. PT emissions
$E_{\rm PT}^{\rm trip}$ use the same distance and a bus/rail factor weighted by
the bus-time share $\omega_{\mathrm{bus},\mathcal D_{\rm sim}}$.

The per-trip externalities $c_1$, $c_{\rm PT}$, and $c_{\rm car}$ are
\begin{align}
c_1&=0,
\nonumber\\
c_{\rm PT}
&=\frac{E_{\rm PT}^{\rm trip}V_{\rm CO_2e}}{10^6},
\label{eq:appnum-pt-externality}\\
c_{\rm car}
&=\frac{E_{\rm car}^{\rm trip}V_{\rm CO_2e}}{10^6}
 +\bar d_{\mathcal D_{\rm sim}}(C_{\rm cong}+C_{\rm other}).
\label{eq:appnum-car-externality}
\end{align}
Here, $V_{\rm CO_2e}$ is the monetary value assigned to one tonne of
CO$_2$e, and $10^6$ converts grams to tonnes. The car externality
$c_{\rm car}$ adds the carbon cost, the marginal congestion cost
$C_{\rm cong}$ of an additional vehicle-kilometre on principal London roads,
and the remaining road costs $C_{\rm other}$. For profile $k$, the expected
externality conditional on not choosing bike, $c_{\neg1,k}$, is
\begin{equation}
c_{\neg1,k}=
\mathbb P_k(2\mid i\neq1)c_{\rm PT}+
\mathbb P_k(3\mid i\neq1)c_{\rm car}.
\label{eq:appnum-nonbike-cost}
\end{equation}

The generalized social-cost objective $\tilde f^{\tilde v}(y)$ in
Eq.~\eqref{eq:objective_function} weights expected incentive expenditure by
$\beta$. To set the policy weight $\beta$, we anchor the optimal incentive of
the car-oriented profile $k=2$ to an incentive implemented in a cycling field
experiment. Ciccone, Fyhri, and Sundf{\o}r paid participants
$R_{\rm Cic}=2.00$ Norwegian kroner (NOK) for each kilometre cycled in their
pay-per-kilometre treatment\fnref{fn:Ciccone2021Cycling}. Applying this rate to the mean trip
distance $\bar d_{\mathcal D_{\rm sim}}=5.41$ km/trip gives
\begin{equation}
y_{\rm Cic}^{\rm NOK}=R_{\rm Cic}\bar d_{\mathcal D_{\rm sim}}
=10.81\ \mathrm{NOK/bike\ trip}.
\label{eq:appnum-ciccone-nok}
\end{equation}
The annual-average exchange rate $X_{\rm NOK/GBP}=11.83$ NOK/\pounds converts
this amount to
\begin{equation}
y_{\rm Cic}=\frac{y_{\rm Cic}^{\rm NOK}}{X_{\rm NOK/GBP}}
=\pounds0.91\ \mathrm{per\ bike\ trip}.
\label{eq:appnum-ciccone-gbp}
\end{equation}

Let $p_2(y):=\mathbb P_1^{u_2}(y)$ and
$\Delta c_2:=c_{\neg1,2}-c_1$. For profile $k=2$, differentiating the
objective $\tilde f^{u_2}(y)$ in Eq.~\eqref{eq:objective_function}, and using
$p_2'(y)=p_2(y)[1-p_2(y)]/\theta$ from Eq.~\eqref{eq:pk}, gives
\begin{align}
0
&=p_2'(y)\bigl(c_1-c_{\neg1,2}+\beta y\bigr)+\beta p_2(y)
\nonumber\\
&=\frac{p_2(y)}{\theta}
\left\{-\bigl[1-p_2(y)\bigr]\Delta c_2
+\beta\bigl[\theta+\bigl(1-p_2(y)\bigr)y\bigr]\right\}.
\label{eq:appnum-beta-foc}
\end{align}
We impose $y_2^\star=y_{\rm Cic}$ in Eq.~\eqref{eq:appnum-beta-foc}; solving
for $\beta$ yields
\begin{equation}
\beta=
\frac{[1-p_2(y_{\rm Cic})](c_{\neg1,2}-c_1)}
{\theta+[1-p_2(y_{\rm Cic})]y_{\rm Cic}}.
\label{eq:appnum-beta}
\end{equation}
The values issued from the calibration are the common scale
$\widehat\theta=\pounds3.68$ and the profile-2 parameters entering
$u_2(i,y)$ in Eq.~\eqref{eq:appnum-systematic-utility}. The payment
$y_{\rm Cic}$ is instead constructed in
Eqs.~\eqref{eq:appnum-ciccone-nok}--\eqref{eq:appnum-ciccone-gbp}. Substituting
$\widehat\Theta$, $\mathbf x_{i,2}$, and $y_{\rm Cic}$ into
Eq.~\eqref{eq:appnum-systematic-utility}, and then evaluating
Eq.~\eqref{eq:logit}, gives $p_2(y_{\rm Cic})=4.76\%$. Thus,
$p_2(y_{\rm Cic})$ is a post-calibration model evaluation, not a calibrated
parameter. The context construction in Eqs.~\eqref{eq:appnum-pt-externality}--
\eqref{eq:appnum-nonbike-cost} gives $c_{\neg1,2}=\pounds2.67$ and $c_1=0$.
Equation~\eqref{eq:appnum-beta} therefore gives $\beta=0.56$.

\subsubsection{Optimal incentives and induced mode choices}

After calibration, the profile-specific behavioral parameters
$\{\xi_k^{\mathrm{bike}},\xi_k^{\mathrm{PT}},\iota_k,
\zeta_k^{\mathrm{lic}},\zeta_k^{\mathrm{own}}\}_{k\in\mathcal K}$, the common
parameters $\widehat\nu$ and $\widehat\theta$, and the representative
attributes $\mathbf x_{i,k}$ determine the systematic utilities $u_k(i,0)$ via
Eq.~\eqref{eq:appnum-systematic-utility}. The externalities $c_1$ and
$c_{\neg1,k}$ and the policy weight $\beta$ then determine the optimal
incentive $y_k^\star$ through Prop.~\ref{prop:y_star}. For each profile
$k\in\mathcal K$ and mode $i\in\{1,2,3\}$, define
\begin{equation}
P_{ik}^{0}:=\mathbb P_i^{u_k}(0),
\qquad
P_{ik}^{\star}:=\mathbb P_i^{u_k}(y_k^\star).
\label{eq:appnum-policy-probabilities}
\end{equation}
The probabilities $P_{ik}^{\star}$ represent the choices generated by the
calibrated profile $k$ if the regulator offers $y_k^\star$. Table
\ref{tab:appnum-policy-effects} reports $P_{ik}^{0}$, $P_{ik}^{\star}$, and
the resulting change in bike probability
$\Delta P_{1k}:=P_{1k}^{\star}-P_{1k}^{0}$.

\begin{table*}[t]
\centering
\footnotesize
\setlength{\tabcolsep}{4.0pt}
\caption{Profile-specific optimal incentives and induced mode choices. For $i\in\{1,2,3\}$, denoting bike, PT, and car, $P_{ik}^{0}=\mathbb P_i^{u_k}(0)$ and $P_{ik}^{\star}=\mathbb P_i^{u_k}(y_k^\star)$. The change $\Delta P_{1k}=P_{1k}^{\star}-P_{1k}^{0}$ is reported in percentage points.}
\label{tab:appnum-policy-effects}
\begin{tabular}{@{}crrrrrrrrr@{}}
\toprule
$k$ & $\pi_k$ & $P_{1k}^{0}$ & $P_{2k}^{0}$ & $P_{3k}^{0}$
& $a_k-b_k$ & $c_{\neg1,k}$ & $y_k^\star$ & $P_{1k}^{\star}$
& $\Delta P_{1k}$\\
 & (\%) & (\%) & (\%) & (\%) & (\pounds) & (\pounds/trip)
 & (\pounds/bike trip) & (\%) & (pp)\\
\midrule
0 & 21.26 & 52.44 & 19.12 & 28.44 & 0.36 & 2.88 & 0 & 52.44 & 0\\
1 & 39.36 & 3.96 & 52.08 & 43.96 & -11.73 & 2.22 & 0.15 & 4.11 & 0.15\\
2 & 39.38 & 3.75 & 43.10 & 53.15 & -11.93 & 2.67 & 0.91 & 4.76 & 1.01\\
\bottomrule
\end{tabular}
\end{table*}

\subsubsection{Monte Carlo simulation}

The experiment consists of $N_{\mathrm{MC}}=100$ independent Monte Carlo
replications, each with horizon $T=18$ timeslots. We select $T=18$ to examine
the initial learning phase, in which the regulator has fewer than twenty
choices from which to infer the adopted profile. 

For replication $m\in\{1,\ldots,N_{\mathrm{MC}}\}$ and timeslot
$t\in\{1,\ldots,T\}$, a common random number
$U_{m,t}\sim\mathrm{Unif}(0,1)$ is generated. The honest and manipulative mode
choices are obtained by applying the inverse cumulative distribution of their
respective choice probabilities $\mathbb P_i^{u_k}(y^t)$ in
Eq.~\eqref{eq:logit} to the same $U_{m,t}$. Pairing the same $U_{m,t}$ removes
variation caused solely by using different random numbers when the two
trajectories are compared.

\begin{table}[t]
\centering
\small
\setlength{\tabcolsep}{4pt}
\caption{Settings of the Monte Carlo experiment.}
\label{tab:appnum-simulation-settings}
\begin{tabular}{@{}lll@{}}
\toprule
Quantity & Symbol & Value\\
\midrule
Simulation horizon & $T$ & 18 timeslots\\
Monte Carlo replications & $N_{\mathrm{MC}}$ & 100\\
Confidence level & $p_{\mathrm{CI}}$ & 0.95\\
Common random numbers & $U_{m,t}$ & paired trajectories\\
\bottomrule
\end{tabular}
\end{table}

Under InstPE, declaring profile~$r$ immediately determines the
offered incentive~$y_r^\star$, while choices continue to follow the
true profile-$0$ distribution.
Under RealPE, an agent adopting profile~$r$ chooses among bike,
public transport, and car according to the multinomial-logit
probabilities of profile~$r$.
Recall that, for online profile inference, the regulator uses only the
binary observation
\[
z^s=\mathbf{1}\{i^s=1\},
\]
so that a bike choice is recorded as~$z^s=1$, while both public
transport and car choices are recorded as~$z^s=0$.
Before timeslot~$t$, the regulator evaluates the Bernoulli
log-likelihoods~$\ln L_k^t(\mathbf h^{t-1})$ defined in
Eq.~\eqref{eq:loglikelihood} and selects the inferred
profile~$\hat k^t$ according to Eq.~\eqref{eq:argmax-logL}.
If several profiles attain the same maximum log-likelihood, the
profile with the smallest optimal incentive~$y_k^\star$ is selected;
if a tie remains, the smallest profile index is selected.

Let $\Delta U_{r,m}^{s,\mathrm{PE}}$ denote the realized manipulation advantage
in timeslot $s$ of replication $m$ under
$\mathrm{PE}\in\{\mathrm{InstPE},\mathrm{RealPE}\}$. The estimated
average manipulation advantage through timeslot $t$ is
\begin{equation}
\widehat{\overline{\Delta U}}_{r,t}^{\mathrm{PE}}
=\frac{1}{N_{\mathrm{MC}}}\sum_{m=1}^{N_{\mathrm{MC}}}
\frac{1}{t}\sum_{s=1}^{t}\Delta U_{r,m}^{s,\mathrm{PE}}.
\label{eq:appnum-average-advantage}
\end{equation}
The 95\% confidence interval for
$\widehat{\overline{\Delta U}}_{r,t}^{\mathrm{PE}}$ is computed from the
$N_{\mathrm{MC}}$ replication-specific time averages using a Student-$t$
critical value. Figure~\ref{fig:appnum-manipulation-advantage} shows
$\widehat{\overline{\Delta U}}_{r,t}^{\mathrm{PE}}$ over the complete horizon.

\begin{figure*}[t]
\centering
\includegraphics[width=0.98\textwidth]{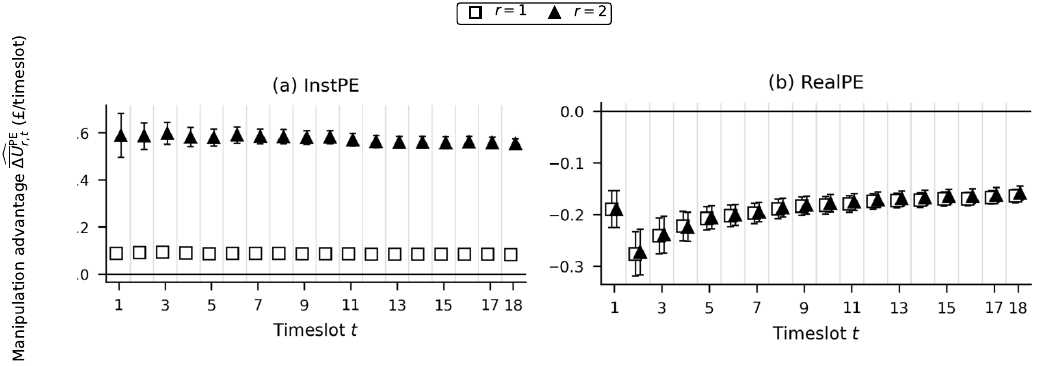}
\caption{Estimated time-averaged manipulation advantage
$\widehat{\overline{\Delta U}}_{r,t}^{\mathrm{PE}}$ from
Eq.~\eqref{eq:appnum-average-advantage}. Vertical intervals are 95\%
Student-$t$ confidence intervals across the Monte Carlo replications.}
\label{fig:appnum-manipulation-advantage}
\end{figure*}

The time-averaged gain from manipulation and cost of manipulation are obtained
by replacing $\Delta U_{r,m}^{s,\mathrm{PE}}$ in
Eq.~\eqref{eq:appnum-average-advantage} with $\textit{GM}_{r,m}^{s,\mathrm{PE}}$
and $\textit{CM}_{r,m}^{s,\mathrm{PE}}$, respectively. Their difference equals
the manipulation advantage by Eq.~\eqref{eq:decomposition}. Figure
\ref{fig:manipulation} and
Table~\ref{tab:appnum-manipulation} report the resulting decomposition at
$T=18$.

\begin{table}[t]
\centering
\footnotesize
\setlength{\tabcolsep}{3.8pt}
\caption{Manipulation results at $T=18$. The estimate
$\widehat{\overline{\Delta U}}_{r,T}^{\mathrm{PE}}$ and its 95\% interval are
in pounds per timeslot; cumulative quantities are in pounds.}
\label{tab:appnum-manipulation}
\begin{tabular}{@{}lcrr@{}}
\toprule
PE & $r$ & $\widehat{\overline{\Delta U}}_{r,T}^{\mathrm{PE}}$ & 95\% interval\\
\midrule
InstPE & 1 & 0.08 & $[0.08,0.09]$\\
InstPE & 2 & 0.55 & $[0.53,0.57]$\\
RealPE & 1 & -0.16 & $[-0.18,-0.15]$\\
RealPE & 2 & -0.16 & $[-0.17,-0.15]$\\
\bottomrule
\end{tabular}

\medskip
\begin{tabular}{@{}lcrrr@{}}
\toprule
PE & $r$ & $\sum_{s=1}^{T}\textit{GM}_r^s$
& $\sum_{s=1}^{T}\textit{CM}_r^s$ & $\sum_{s=1}^{T}\Delta U_r^s$\\
\midrule
InstPE & 1 & 1.47 & 0 & 1.47\\
InstPE & 2 & 10.04 & 0 & 10.04\\
RealPE & 1 & 0.36 & 3.34 & -2.99\\
RealPE & 2 & 0.49 & 3.36 & -2.87\\
\bottomrule
\end{tabular}
\end{table}

For each preference-elicitation model $\mathrm{PE}$ and simulated profile $r$,
define the analytic long-run average manipulation advantage and its finite-
horizon discrepancy as
\begin{align}
\overline{\Delta U}_{r,\infty}^{\mathrm{PE}}
&:=\lim_{t\to\infty}\frac{1}{t}\sum_{s=1}^{t}
\Delta U_r^{s,\mathrm{PE}},
\label{eq:appnum-long-run-advantage}\\
\Gamma_{r,T}^{\mathrm{PE}}
&:=\left|\widehat{\overline{\Delta U}}_{r,T}^{\mathrm{PE}}
-\overline{\Delta U}_{r,\infty}^{\mathrm{PE}}\right|.
\label{eq:appnum-finite-horizon-gap}
\end{align}
For InstPE with $r=1,2$, the values of
$\overline{\Delta U}_{r,\infty}^{\mathrm{InstPE}}$ are 0.08 and 0.56
\pounds/timeslot; for RealPE with $r=1,2$, the values of
$\overline{\Delta U}_{r,\infty}^{\mathrm{RealPE}}$ are $-0.17$ and $-0.13$
\pounds/timeslot. The corresponding discrepancies
$\Gamma_{r,T}^{\mathrm{PE}}$ at $T=18$ are 0.00, 0.00, 0.00, and 0.03
\pounds/timeslot. Thus, the finite-horizon estimates have the same sign as all
four analytic long-run averages; the largest remaining finite-horizon
discrepancy is the RealPE case with $r=2$.

\subsection{Social-Impact Relevance}
\label{sec:deployment-interpretation}

The proposed framework is intended as a pre-deployment stress test
for personalized incentive policies. Before implementing such a
policy, a regulator can calibrate behavioral profiles from observed
choices, compute the incentives associated with these profiles, and
evaluate whether the additional incentive that an agent may obtain
through manipulation compensates for the resulting Cost of
Manipulation. The theoretical bounds provide an analytical
robustness certificate, while the simulation procedure evaluates the
same mechanism under calibrated behavioral heterogeneity and
stochastic choices.

The London case study provides empirical grounding for this
evaluation. The choice model and behavioral profiles are calibrated
from observed travel-mode choices, after which the interaction
between preference elicitation, personalized incentives, and
strategic behavior is simulated under controlled and reproducible
conditions. This design isolates the mechanism studied in the paper:
whether repeatedly imitating another behavioral profile remains
advantageous once the regulator learns from the agent's choices. The
results show that the conclusion obtained under instantaneous
preference elicitation can be reversed when this learning process is
explicitly represented. In the calibrated instance, realistic
preference elicitation makes the Cost of Manipulation exceed the Gain
from Manipulation while the incentive policy continues to increase
the share of sustainable choices.

This use of calibrated simulation is appropriate for policy
screening: it allows candidate policies to be rejected or revised
before exposing individuals or public budgets to an insufficiently
robust intervention. It also makes the assumptions governing the
policy explicit and auditable. In particular, the regulator applies
a deterministic and pre-announced mapping from inferred profiles to
incentives, and the preference-elicitation procedure can operate on
the binary information required to determine whether the
incentivized alternative was selected, rather than requiring the
complete observation of every non-incentivized choice.

The analysis therefore supports a concrete deployment workflow:
calibrate the behavioral model for the population of interest,
evaluate manipulation incentives under realistic preference
elicitation, and deploy the policy only when the resulting
manipulation advantage remains non-positive over the relevant
interaction horizon. When this condition is not satisfied, the same
framework identifies the profiles and time horizons for which the
policy should be redesigned. The contribution is thus not tied to a
single transport intervention: it provides a reproducible procedure
for assessing the strategic robustness of personalized
sustainability policies before implementation.


\end{document}